\documentclass[letter]{article}

\usepackage{lineno,hyperref}
\usepackage{enumitem}

\usepackage{mathtools,amssymb,amsthm}

\DeclarePairedDelimiter\floor{\lfloor}{\rfloor}

\usepackage{array}

\usepackage{algorithm}
\usepackage{rotating}
\usepackage{algpseudocode}
\algtext*{EndWhile}
\algtext*{EndIf}
\algtext*{EndFor}
\algtext*{EndFunction}

\usepackage[caption=false,font=footnotesize,labelfont=sf,textfont=sf]{subfig}

\usepackage{dblfloatfix}

\usepackage{url}

\usepackage{enumitem}

\usepackage{setspace}

\usepackage{tcolorbox}
\tcbuselibrary{breakable}

\newtheorem{theorem}{Theorem}
\newtheorem{corollary}{Corollary}[theorem]
\newtheorem{lemma}{Lemma}
\usepackage{authblk}
\usepackage{booktabs}

\begin{document}

\title{Lightweight Collaborative Anomaly Detection \\for the IoT using Blockchain}

\author{Yisroel Mirsky$^{1,2}$, Tomer Golomb$^{2}$ and Yuval Elovici$^{2}$\\
$^{1}$Georgia Institute of Technology, Atlanta, USA\\
$^{2}$Ben-Gurion University, Beer Sheva, Israel\\
\vspace{.5em}
yisroel@gatech.edu, golombt@post.bgu.ac.il, elovici@bgu.ac.il\\
\vspace{1em}
\small
Preprint of accepted publication, June 2020: \\Journal of Parallel and Distributed Computing, Elsevier, ISSN: 0743-7315  
}

\date{}
\maketitle

\begin{abstract}

Due to their rapid growth and deployment, the Internet of things (IoT) have become a central aspect of our daily lives. 
Unfortunately, IoT devices tend to have many vulnerabilities which can be exploited by an attacker. Unsupervised techniques, such as anomaly detection, can be used to secure these devices in a plug-and-protect manner. 
However, anomaly detection models must be trained for a long time in order to capture all benign behaviors. Furthermore, the anomaly detection model is vulnerable to adversarial attacks since, during the training phase, all observations are assumed to be benign.
In this paper, we propose (1) a novel approach for anomaly detection and (2) a lightweight framework that utilizes the blockchain to ensemble an anomaly detection model in a distributed environment.
Blockchain framework incrementally updates a trusted anomaly detection model via self-attestation and consensus among the IoT devices. 
We evaluate our method on a distributed IoT simulation platform, which consists of 48 Raspberry Pis. The simulation demonstrates how the approach can enhance the security of each device and the security of the network as a whole.
	
\end{abstract}


\section{Introduction}
The Internet of Things (IoT) is the next evolution of the Internet \cite{khan2012future} where devices, of any kind and size, will exchange and share data autonomously among themselves. 
By exchanging data, each device can improve their decision-making processes. IoT devices are ubiquitous in our daily lives and critical infrastructure. For example, air conditioners, irrigation systems, refrigerators, and railway sensors \cite{rail} have been connected to the Internet in order to provide services and share information with the relevant controllers. Due to the benefits of connecting devices to the Internet, massive quantities of IoT devices have been developed and deployed. This has led leading experts  to believe that by 2020 there will be more than 20 billion devices connected to the Internet \cite{eddy2015gartner}.

While the potential for IoT devices is vast, their success depends on how well we can secure these devices. 
However, IoTs are diverse and have limited resources. Therefore, securing them is a difficult challenge which has taken a central stage in both industry and academia.
One significant security concern with IoTs is that many manufactures do not invest in the security of these devices during their development. Furthermore, discovered vulnerabilities are seldom patched by the manufacture \cite{notPatchingIOTs}. These vulnerabilities enable attackers to exploit the IoT devices for nefarious purposes \cite{schneier2014internet} which endanger the users' security and privacy.

There are various security tools for detecting attacks on embedded devices. One such tool is an intrusion detection system (IDS). An anomaly-based IDSs learn the normal behavior of a network or host, and detect when the behavior deviates from the norm. In this way, these systems have the potential to detect new threats without being explicit programmed to do so (e.g., via remote updates). Aside from being able to detect novel `zero-day' attacks, this approach is desirable because there is vertically no maintenance required.

In order to prepare an anomaly-based IDS (or any anomaly detection model), the system must collect and learn from \textit{normal} observations acquired during a time-limited ``training phase''. A fundamental assumption is that the observations obtained during the training phase are both benign and capture all of the device's possible behaviors.
This assumption might hold true in some systems. However, when considering the IoT environment, this assumption is challenging for the following reasons:

\begin{enumerate}

	\item  \textbf{Model Generality} It is possible to train the anomaly detection model safely in a lab environment. However, it is difficult to simulate all the possible deployments and interactions with the device. This is because some logic may be dependent on one or more environmental sensors, human interaction, and event based triggers. This approach is also costly and required additional resources. Alternatively, the model can be trained on-site during the deployment itself. However, the model will not be available for execution (detection of threats) until the training phase is complete. Furthermore, it is questionable whether the trained model will capture benign yet rare behaviors. For example, the behavior of the motion detection logic of a smart camera or the response generated by a smoke detector while sensing a fire. These rare but legitimate behaviors will generate false alarms during regular execution.
	
	\item  \textbf{Adversarial Attacks} 
	Although training on-site is a more natural approach to learning the normal behavior of an IoT device, the model must assume that all observation during the training-phase are benign. This approach exposes the model to malicious observations, thus enabling an attacker to exploit the device to evade detection or cause some other adverse effect.
\end{enumerate}

To overcome these challenges, the IoT devices can collaborate and train an anomaly detection model together. Consider the following scenario:

Assume that all IoT devices of the same type simultaneously begin training an anomaly detection model, based on their own locally observed behaviors. 
The devices then share their models with other devices of the same type. Finally, each device merges the received models, into a single model by filtering out potentially malicious behaviors. Finally, each device uses the combined model as it's own local anomaly detection model. As a result, the devices (1) collectively obtain an anomaly detection model which captures a much wider scope of all possible benign behaviors, and (2) are able to significantly limit adversarial attacks during the training phase. The latter point is because the \textit{initial} training phase is much shorter (scaled according to the number of devices), and rare behaviors unseen by the majority are filtered out.

Using concept, we present a lightweight, scalable framework which utilizes the blockchain concept to perform distributed and collaborative anomaly detection on devices with limited resources.

A blockchain is an innovative protocol for a distributed database, which is implemented as a chain of blocks and managed by the majority of participants in the network \cite{swan2015blockchain}.
Each block contains a list of records and a hash value of the previous block and is accepted into the chain if it satisfies a specific criteria (e.g., bitcoin's proof-of-work criterion \cite{nakamoto2008bitcoin}).
The framework uses the blockchain's concept to define a collaboration protocol which enables devices to autonomously train a trusted anomaly detection model incrementally. The protocol uses self-attestation and consensus among the IoT devices to protect the integrity of the trained model. In our blockchain, a record in a block is a model trained on a specific device, and a block in the chain represents a potential anomaly detection model which has been verified by a significant mass/majority of devices in the system. By using the blockchain as a secured distributed ledger, we ensure that the devices (1) are using the latest validated anomaly detection model, and (2) can continuously contribute to each other's model with newly observed benign behaviors. 

Furthermore, in this paper we also propose a novel approach for performing anomaly detection on a local device using an Extensible Markov Model (EMM) \cite{bhat2008extended}. The EMM tracks a program's jump sequences between regions on the application's memory space. The EMM can be incrementally updated and merged with other models, and therefore can be trained with real-world observations across multiple devices in parallel. Although there are many other methods for modeling sequences, we chose the EMM model because:
\begin{enumerate}
	\item The update and prediction procedures have a complexity of $O(1)$. This is critical considering that many IoT devices have weak processors. 
	\item Our collaborative framework requires a model which can be merged with other models efficiently. Moreover, to filter out malicious transitions during the combine step, we needed an efficient and clear algorithm for comparing learned behaviors between different models. The process of comparing and combining other discreet transitional anomaly detection models can be complex or simply has not been defined. 
	\item In our evaluations, we found that the EMM performs better than other algorithms in our anomaly detection task.
\end{enumerate}

We evaluate both the framework and the anomaly detection model on our own IoT emulation platform, involving 48 Raspberry Pis.
We simulate several different IoT devices to assert that the evaluation results do not depend on the IoT device's functionality. 
Moreover, we exploit real vulnerabilities in order to evaluate our method's capability in detecting actual attacks.
From our evaluations, we found that our method is capable in creating strong anomaly detection models in a short period of time, which are resistant to adversarial attacks.
To encourage further research and development, the reader may download our data sets and source code from GitHub.\footnote{\texttt{https://git.io/vAIvd}.} We have also \href{https://drive.google.com/drive/folders/15gLytEJyQyYCmhB-EZSkES77KsuCW0hw?usp=sharing}{published a blockchain simulator for our protocol} to help the reader understand and implement the work in this paper.\footnote{\texttt{https://github.com/ymirsky/CIoTA-Sim}} 

In summary, this paper's contributions are:
\begin{itemize}
	\item \textbf{A method for detecting code execution attacks by modeling memory jumps sequences} - We define and evaluate a novel approach to efficiently detect abnormal control-flows at a set granularity. The approach is efficient in because we track the program counter's flow between regions of memory, and not actual memory addresses or system calls. As a result, the model is compact (has relatively few states) and is suitable for devices with limited resources (IoT devices).
	
	\item \textbf{A method for enabling safe distributed and collaborative model training on IoTs} - We outline a novel framework and protocol which uses the concept of blockchain to collaboratively train an anomaly detection model. The method is decentralized, reduces train time, false positives, and is robust against potential adversarial attacks during the initial training phase.
\end{itemize}

The rest of the paper is organized as follows. 
In Section \ref{sec:relworks}, we review related work, and discuss how the proposed method overcomes their limitations.
In Sections \ref{sec:anom} and \ref{sec:ciota}, we present introduce our novel host-based anomaly detection algorithm and the framework for applying the algorithm in the collaborative distributed setting using the blockchain.
In Section \ref{sec:eval}, we evaluate the proposed method on several different applications and use-cases, and discuss our insights. In section \ref{sec:security} we analyze the framework's security.
In Section \ref{sec:discussion}, we provide a discussion on the security and challenges of implementing the proposed framework. Finally, in Section \ref{sec:conclusion} we present a summary and conclusion.

\section{Related Works}\label{sec:relworks}
The primary aspects of this work relate to both Intrusion Detection and IoT Security. Therefore, in this section we will discuss recent works from both fields, and the limitations of these approaches.

\subsection{Discreet Sequence Anomaly Detection for Intrusion Detection}
Software inevitable contains flaws which pose security vulnerabilities if exploited by an attacker. Many of these vulnerabilities remain unknown until they are discovered and exploited in the wild (referred to as zero days). An effective way to detect these exploits is to analyze a program's behavior in real-time. 

A program's behavior can be observed during runtime by monitoring its system calls, or by tracking the program in the memory \cite{maske2016advanced,yoon2017learning,kim2016lstm,khreich2017anomaly}. In both cases, the behavior is observed as an ordered sequence of events on which anomaly detection can be performed \cite{ahmed2016survey,chandola2010anomaly}. To detect attacks in these sequences, many works utilize discreet sequence anomaly detection algorithms. We will now summarize these works in chronological order.

In \cite{forrest1996sense} the authors create a database of normal sequences by windowing over system calls, and flag sequences as anomalous sequences if they do not appear in the database. In \cite{kosoresow1997intrusion} the authors extended the windowing approach to longer sequences via partitioning. In \cite{lee1997learning} the authors use RIPPER to extract concise rule sets from systems calls to classify malicious sequences. The authors then expanded their work in \cite{lee1998data} by using the frequent episodes algorithm, computing inter/intra-audit record patterns, and by proposing a general agent architecture. In \cite{hofmeyr1998intrusion} the authors proposed a system based on the defenses of natural immune systems. First a database of short normal sequences is created. Then new sequences are scored according to the number of matches (substrings) the sequence has in common. 

In \cite{warrender1999detecting} the authors performed a comparative evaluation involving Hidden Markov Models (HMM), RIPPER, and threshold-based sequence time delay embedding (t-STIDE). An HMM is similar to a MC except that it model transmissions based on output symbols at each state. We did not use an HMM since the framework needs a light weight model that can be trained efficiently and can be merged with other models. t-STIDE works by looking up the frequency of new sequences (window) in normal dictionary (hash table). Infrequent sequences below a given threshold are considered anomalous. The authors found that the HMM provided the best performance, but t-STIDE had similar performance and was significantly faster to train.

In \cite{michael2000two} the authors propose modeling a finite-state machine over system calls such that novel sequences are labeled anomalous. In \cite{gao2002hmms} the authors apply an HMM over a sliding window of system calls. In \cite{tandon2003learning} the authors revisit the use of association rule mining by considering the system call's arguments. Based on a mining algorithm called LERAD, the authors propose three variants which out performed t-STIDE for certain attacks. In \cite{hoang2003multi} the authors propose a multi-layer approach which first uses a normal database and then passes suspicious sequences to a HMM for further analysis. In \cite{yeung2003host} the authors use a Markov Chain to model normal shell-command sequences, and then detect abnormal (malicious) sequences as an indication of misuse in the system.

In \cite{eskin2001modeling,mazeroff2003probabilistic,mazeroff2008probabilistic} the authors use probability suffix trees (PST) to model normal system call sequences. A PST is a variable length Markovian model which forms a tree-liek data structure.

In \cite{hu2009simple} the authors propose a method for speeding up HMM training on system calls by 50\%. They accomplish this by prepossessing the training sequences and by performing incremental training. In \cite{xie2013evaluating} the authors prose a kernel trick to transfer sequences to Euclidean space in order to perform kNN lookups. In \cite{kim2016lstm} the authors propose the use of a long-term short-term (LSTM) neural networks to detect abnormal system call sequences. In \cite{chawla2018host} the authors extent the work of \cite{chawla2018host} by stacking convolutional networks (CNN) followed by a recurrent neural network (RNN) with Gated Recurrent Units (GRU). Although the use of GRU reduced the training time, the authors needed to use powerful GPUs to train their network. 

Our proposed framework uses an EMM, a type of Markov Chain, as the anomaly detection model for detecting abnormal control flows in the memory of applications. In contrast to the above works, the limitation to these approaches are:

\begin{description}
	\item[Attack Vector Coverage] Many exploits do not use the shell or require the evocation of abnormal system-calls, so the Markov model would not observe any malicious activity during their exploitation processes. For example, exploitation of a buffer overflow vulnerability can be accomplished without making explicit calls. 
	\item[Modeling the True Behavior] An application's system and shell calls only capture an application's high-level behavior. As a result, some exploits can be designed so that the executed code will generate seemingly benign sequences (obfuscation) and evade detection. Furthermore, some malware may only require to make benign call sequences to accomplish its objective. For example, a randsomware will read and write files via system-calls (benign) but encrypt the files internally (malicious).
	\item[System Overhead] In order to intercept the system calls of a specific application, one must intercept the system-calls of all applications. As a result, these approaches are suitable for devices with strong computational power such as personal computers, but not IoTs. Moreover, models such as HMMs and neural networks cannot be trained (and sometimes not even executed) on IoTs.
\end{description}

By modeling a Markov Chain on a target application's general jumps through its memory space, our approach is not restricted by the above limitations. Namely, our approach can (1) capture the internal behavior of the application, regardless of the system-calls or shell-code, (2) detect exploitation of vulnerabilities occurring within an application's memory space, and (3) be applied to specific applications, as opposed to the whole system, thus minimizing overhead --making our approach appropriate for IoTs. 

Similar to our approach, in \cite{7167219} the authors detect anomalous activities by maintaining a heatmap of the \textit{kernel's} memory space. An anomaly is detected when the probability of a region of the kernel's memory being accessed is below a threshold. By doing so, the authors were able o detect abnormal application activities reflected by interactions with the kernel. This work differs from ours in the following ways:
\begin{enumerate}
	\item The kernel-heatmap method cannot detect all of those which our method can. For example, code reuse attacks are ignored because the kernel interactions seem normal. Moreover, abnormal interactions with the kernel can be considered benign because \textit{other} applications may be performing similar interactions. For example, when privilege escalation is obtained and abused, restricted system calls will not seem abnormal because the context of the requesting app is not considered.
	\item When an anomaly is detected, there is indication of which application has been compromised. This makes it harder to mitigate the threat.
	\item The method in \cite{7167219} suffers from significantly higher false alarm rates than an EMM. This is because the probability of accessing a memory region is normalized over all accesses. Therefore, rare benign memory interactions are considered anomalous. In contrast, by using an EMM over memory regions, we consider the transition across the memory space which provides an implicit context for each interaction. Later in section \ref{sec:eval} we provide a comparative evaluation.
	\item Like all other anomaly detection algorithms (including EMMs), the method in \cite{7167219} is subject to adversarial attacks (poisoning) during training and false positives due to rare benign behaviors (due to human interactions and other stimuli). In our paper we propose a framework which provides accelerated on-site model training in a hostile environment via collaboration, filtration, and self-attestation. 
\end{enumerate}

Another approach to deploying an IDS is to distribute the detection across multiple devices \cite{abraham2007d, zhang2011distributed, snapp1991dids}. In these approaches, the devices share information with one another regarding malicious traffic and the network's state. Similar to our method, a distributed IDS utilizes on collaboration between devices. However, the proposed methods are limited to analyzing network traffic. In many cases, network traffic from a device cannot indicate the exploitation of an application running on the device (e.g., encrypted payloads). When considering the IoT topology and the vision of allowing them to autonomously exchange data, a network based IDS might be problematic, since network traffic near each IoT device may differ significantly. In contrast, our anomaly detection approach on an application's the memory jumps is not affect by the diversity of network traffic near each device. Furthermore, distributed IDS solutions are designed to work collectively as a single intrusion detection system. However, should one node be compromised by an attacker, the security of the entire system may fail. In contrast, our method allows for safe collaboration via self attestation and model anomaly filtration --which makes compromising the whole system much more difficult.

\subsection{IoT Specific Solutions}
The IoT device security solutions have been researched extensively over the last few years.
However, the proposed solutions typically do not address all of an IoT's characteristics: their (1) mass quantity, (2) limited resources, (3) global deployment, (4) dependence on external sensors/triggers.

In \cite{huuck2015iot,oh2014malicious} the authors propose deploying static analysis tools on the IoT devices. However, these approaches require that (1) the device maintain a database of virus signatures and (2) that experts continuously update this database. Furthermore, these approaches are not sufficient when facing viruses which can only be detected during runtime (e.g., execution of a malicious encrypted payload). Our method is anomaly-based and therefore can detect threats automatically without human intervention, and performs continuous dynamic analysis of an application's behavior.

Several studies try to secure IoT devices by deploying an anomaly detection model on the device itself \cite{raza2013svelte,arrington2016behavioral,o2014anomaly, taneja2013analytics}. Some of them suggest to simply apply traditional solutions (meant for stronger devices), while others suggest a novel approaches which are more light-weight. A common denominator for all these approaches is: they neglect of the fact that (1) an anomaly detection model is sensitive to adversarial attacks during the training phase, and (2) rare benign activities (which did not appear in the initial training data) can generate false positives (e.g., an IoT smoke detector being triggered).
Our method, on the other hand, has a very short initial training-phase and learns from the experiences (events) of millions of IoTs.

Other studies propose that a centralized server should be deployed \cite{abera2016c,jager2017rolling,ott2015trust}. However, the centralized approach does not scale well with the number of IoT devices. Our method is distributed and autonomous.

Another direction in the literature is to deploy a network-based IDS at the gateway of IoT distributions \cite{taneja2013analytics,chen2011novel}. Although this is a suitable solution for smart homes and offices, it does not scale to industrial deployments (e.g., smart railways), or where the IoT devices are connected directly to the Internet (e.g., some survallaince cameras). Our method does not depend on the IoT devices' deployment or topology.

Other studies have tried to avoid the issue of training altogether, by using a trust anchor, such as an IoT device's functional relationship to detect anomalies. In \cite{moon2015functional}, the authors propose executing every distributed computation twice across different IoT devices and then compare the results to detect deviations (infected devices). However, this method was only designed to protect specific types of IoT devices, from specific types of attacks. Our method is generic to the type of device, and the type of attack.

Other trust anchors solutions include the Trusted Platform Module (TPM) \cite{morris2011trusted} and Trusted Execution Environment (TEE) \cite{yiu2015armv8}. ARM's TrustZone \cite{su2011multi} is a TEE implemented in the hardware, providing a one-way separation between two worlds: ``unsecured'' and ``secured''. In \cite{abera2016c} the authors proposed C-FLAT which utilizes the Trust Zone for attesting the IoT device's control-flow behavior against a simulation run in parallel on a central server. Although an application's control-flow can be used to detect a vast range of code execution attacks, C-FLAT is limited to specific IoT devices which (1) do not execute code continuously or (2) devices whose behavior is not affected by external sensory events (e.g., smart cameras). Our method analyzes control-flow behavior to detect abnormalities dynamically on-site, and therefore does not have these limitations.

\section{The Anomaly Detection Model}\label{sec:anom}
In this section, we present a novel method for efficiently modeling an application's control-flow, and then detecting abnormal patterns with the trained model. The method is applied locally and continuously on a single IoT device. Later, in Section \ref{sec:ciota}, we will present the proposed framework for enabling the decentralized collaborative training of the anomaly detection model.

\subsection{Motivation}
When an application is executed, the kernel designates a region of memory for the program to operate in. The region contains the program's code (machine instructions) and room for data (e.g., variables) to be manipulated by the program. As a program runs, a program counter (PC) tracks the current location (in memory) of the current instruction being executed. The PC will jump to different locations when functions, if statements, and loops are performed. By following the location of the PC over the application's region in memory, a pattern emerges. This pattern captures the behavior (control-flow) of the application. The objective is to model the normal behavior of an application's control flow, and then later detect when the behavior changes. 

When an attacker does not have the victim's credentials, the attacker may attempt to exploit a software vulnerability in order to obtain access to restricted assets, or to perform some other undesirable task (e.g., install a bitcoin mining bot). When an exploit is executed on an application, the control-flow of the app will deviate from the behavior intended by the app's developers. By detecting this abnormality, we can identify the threat and then take the proper steps to alter the user and mitigate it. 

Buffer overflow and code-reuse are examples of attacks which abnormally affect the PC's location in memory. Another example is the ``Zimperlich'' \cite{Zimperlich} vulnerability in Android which gives the attacker privileged escalation. When exploited, the ``Zimperlich'' causes the \textit{setuid} operation to fail. 
However, by monitoring the control-flow of the application, we can detect that the app was attempting access to the region of memory where \textit{setuid} is located, at an unusual time. As a result, we can raise an alert which will reveal the attack to the user.

With this approach, it is challenging for the attacker to evade detection. This is because most systems cannot change the code loaded into memory. Therefore, in order for the attacker to execute code which will hide the malicious activities, the attacker must either (1) add code of his own (which will make the PC jump), or (2) override existing code with his own (which will change the behavioral flow of the PC). This places the attacker in a \textit{catch-22}, where his exploit will ultimately detected as an anomaly (Fig. \ref{ExecutionFlow}).

\begin{figure}[!t]
	\centering
	\includegraphics[width=.7\columnwidth]{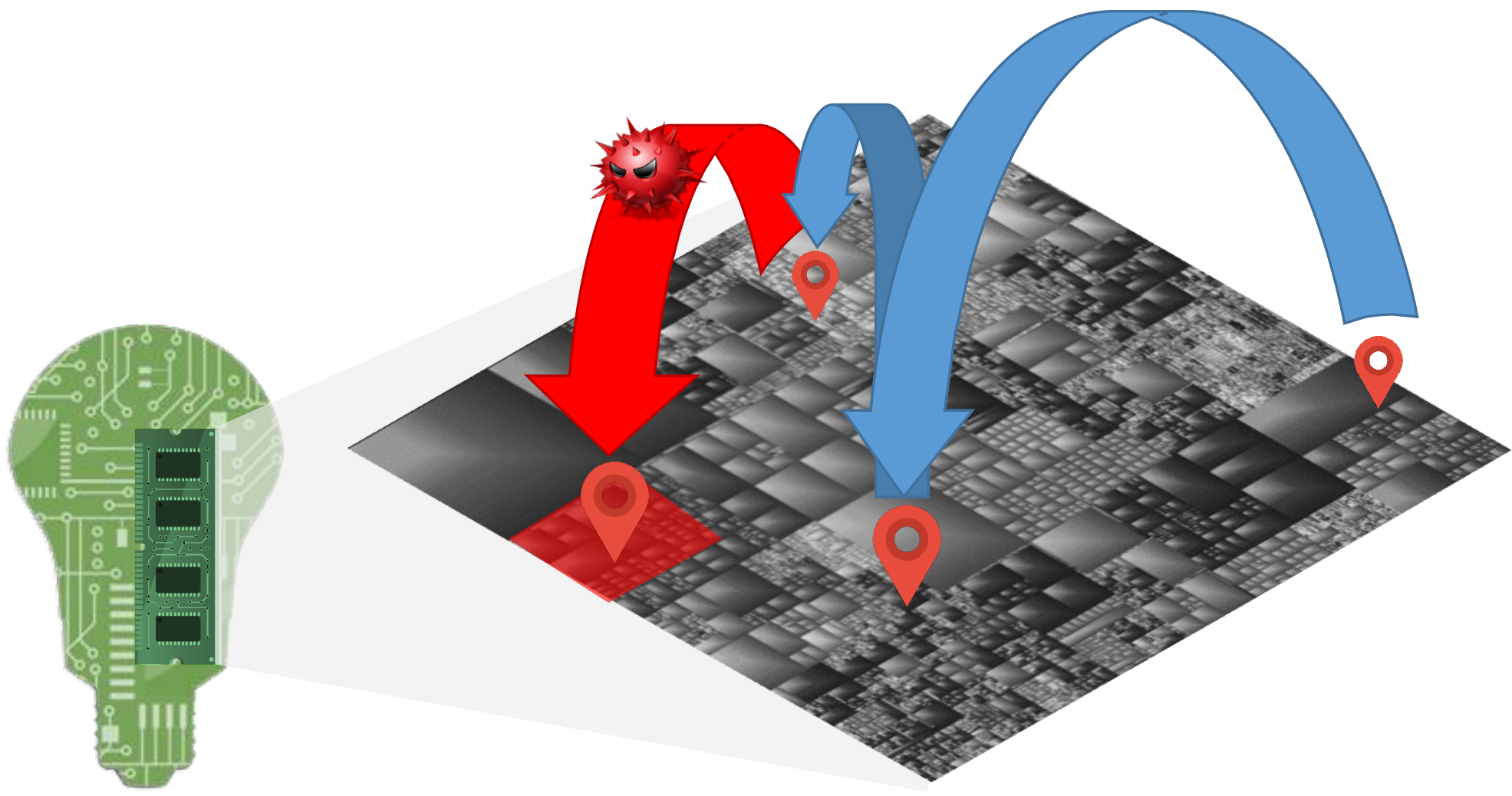}
	\caption{A visualization of a smart-light's program control-flow over the memory space, and the affect caused when a vulnerability is exploited to run malicious code.}
	\vspace{-0.3cm}
	\label{ExecutionFlow}
\end{figure}

\subsection{Markov Chains}
In order to efficiently model sequences, we use a probabilistic model called a Markov chain (MC). An MC is a {\em memory-less process}, i.e., a process where the probability of transition at time $t$ only depends on the state at time $t$ and not on any of the states leading up that state.
Typically, an MC is represented as an adjacent matrix $M$, such that $M_{ij}$ stores the probability of transitioning from state $i$ to state $j$ at any given time $t$. Formally, if $X_t$ is the random variable representing the state at time $t$, then 
\begin{equation} \label{eq:3}
M_{ij}=Pr(X_{t+1}=j|X_{t}=i)
\end{equation}
An EMM \cite{bhat2008extended} is the incremental version of the MC. Let $N=[n_{ij}]$ be the frequency matrix, such that $n_{ij}$ is the number of transitions which have occurred from state $i$ to state $j$. From here, the MC can be obtained by
\begin{equation}\label{eq:emm}
M=[M_{ij}]=\left[\frac{n_{ij}}{n_i}\right]
\end{equation}
where $n_i=\sum_j n_{i,j}$ is the total number of outgoing transitions observed by state $i$. By maintaining $N$, we can update $M$ incrementally by simply adding the value '1' to $N_{ij}$ whenever a transition from $i$ to $j$ is observed. In most cases, $N$ is a sparse matrix (most of the entries are zero). When implementing an EMM model, one can use efficient data structures (e.g., compressed row storage or hash maps) to track large numbers of states with a minimal amount of storage space.

If $N$ was generated using normal data only, then the resulting MC can be used for the purpose of anomaly detection \cite{Patcha20073448}. Let $Q_{k}$ be the last $k$ observed states in the MC, where $k\in\{1,2,3,\ldots\}$ is a user defined parameter. 
The simplest anomaly score metric is the probability of the observed trajectory $Q_{k}=(s_0,\ldots,s_t)$ w.r.t the given MC. This is given by
\begin{equation}
Pr(Q_{k})=Pr(\bigwedge_{i=0}^k (X_i=s_i))=\prod_{i=0}^{k-1} M_{s_i,s_{i+1}}
\label{eq:trajectory-prob}
\end{equation}
When a new transition from state $i$ to state $j$ is observed, we assert that the transition was anomalous if $Pr(Q_k)<p_{thr}$, where $p_{thr}$ is a user defined cut-off probability. However, for large $k$, or in the case of noisy data, $Pr(Q_k)$ can generate many false positives. In this case, the average probability of the sequence can be used 
\begin{equation}
\overline{Pr}(Q_{k})=\frac{1}{t}\sum_{i=0}^{k-1} M_{s_i,s_{i+1}}
\label{eq:trajectory-avprob}
\end{equation}
Lastly, to avoid corrupting the model, one should not update $N_{ij}$ with a transition deemed an abnormal (part of an attack).

\subsection{Detecting Abnormal Control-Flows in Memory}
To track the logical address of the PC in real-time, we can use a kernel debugger such as the Linux performance monitoring API (\texttt{linux/perf\_event.h})\footnote{The API can be found here: \texttt{http://man7.org/linux/man-pages/\\man2/perf\_event\_open.2.html}, and sample code can be found here: \texttt{https://git.io/vAIvd}}. The debugger runs in parallel to the target application and tracks addresses of the PC. The debugger periodically reports the addresses observed since the last report. The sequence of observed addresses can then be modeled in the MC.

Modeling a memory space as states in an MC is a challenging task. Due to memory limitations, it is not practical to store every address in memory as a state. Doing so would also require us to track the location of the PC after every operation. This would incur a significant overhead. Therefore, we propose that a state in the MC should be region of memory (Fig. \ref{fig:memflow}). We also configure the debugger to report right after \texttt{branch} and \texttt{jump} instructions. To accomplish this, we used a kernel feature via the debugger.\footnote{The feature is called \textit{coresight} in ARM CPUs, and \textit{lastjump} in Intel CPUs.}

Let $PC_{addr}$ be the current logical address of the application's program counter, where \texttt{0x0} is the start of the app's logical memory-space. Let $B$ be the partition size in Bytes, defined by the user. The MC state $i$, in which the program is currently located, is obtained by
\begin{equation}
i = \floor*{\frac{PC_{addr}}{B}}
\end{equation}
The partition size of a memory space is a user defined parameter. When selecting the partition size, there is a trade off between the true positive rate and memory requirements. For an Apache web-server, we found that a partition size of 256 Bytes was is enough to detect our evaluated attacks, where the memory consumption of the model $N$ was only 20KB of memory.

In Algorithm \ref{alg:anomDetect}, we present the complete process for modeling and monitoring an application's control-flow in memory. There, $T_{grace}$ is the initial learning time given to the MC, before we start searching for anomalies. 

\begin{figure}[!t]
	\centering
	\includegraphics[width=.7\columnwidth]{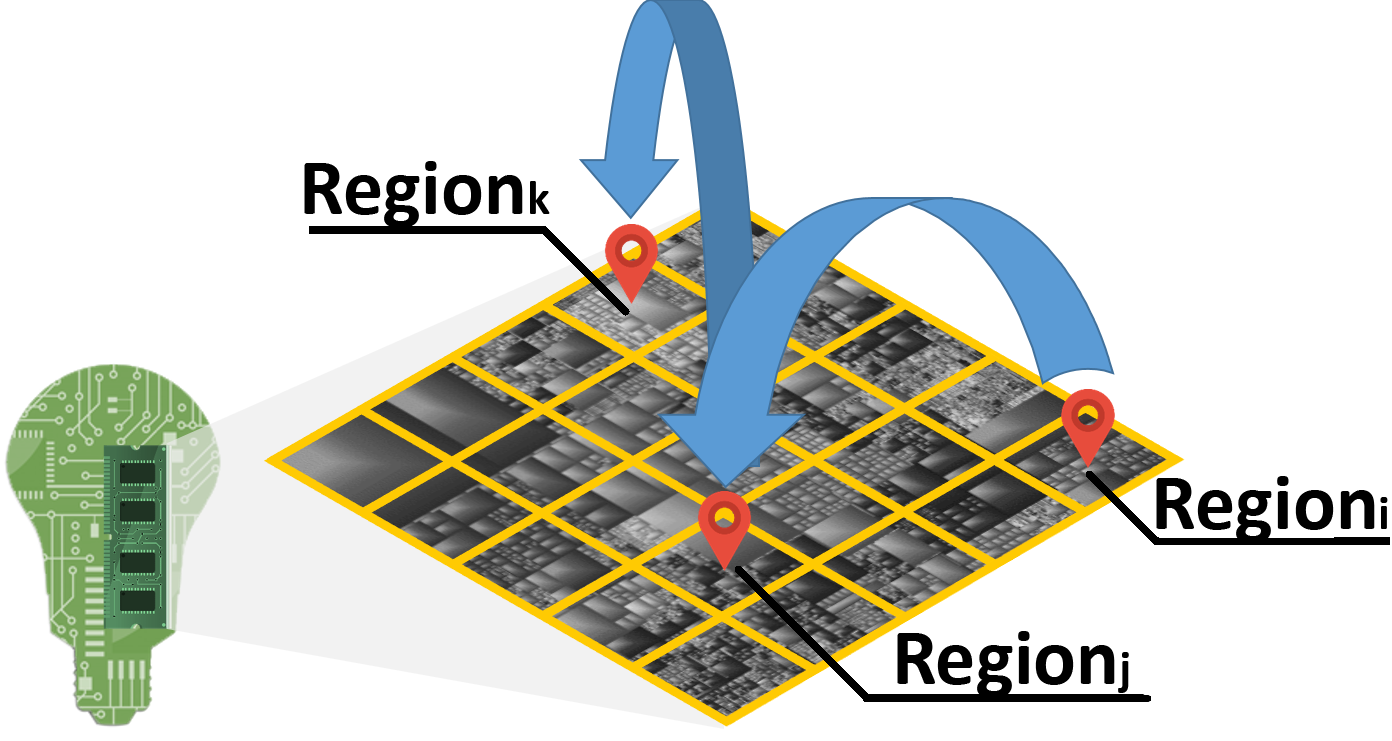}
	\caption{A visualization of partitioning a smart-light's memory-space into states in a Markov Chain.}
	\label{fig:memflow}
\end{figure}

\begin{algorithm}\label{alg:monitor}
	\caption{The algorithm for training an MC, and detecting anomalies in an application's control flow via the memory.}
	\label{alg:anomDetect}
	\begin{algorithmic}[1]
		\Function{Monitor}{app\_name, $k$, $B$, $p_{thr}$, $T_{grace}$}
		\State $N \leftarrow DynamicSparseMatrix()$ \Comment{init MC}
		\State $Q_k \leftarrow FIFO(k)$ \Comment{init state trace ring-buffer}
		\State $i \leftarrow 0$ \Comment{initial state}
		\State fd $\leftarrow$ register(app\_name) \Comment{track app with debugger}
		\While{buffer $\leftarrow$ read(fd)}
		\For{$addr \in$ buffer}
		\State $j = \floor*{\frac{PC_{addr}}{B}}$ \Comment{determine current state}
		\If{notInGrace($T_{grace}$) \textbf{and} $\overline{Pr}(Q_{k})<p_{thr}$}
		\State raise alert
		\Else
		\State $N_{ij}++$\Comment{update MC}
		\EndIf	
		\State $i = j$	
		\EndFor
		\EndWhile
		\EndFunction
	\end{algorithmic}
	
\end{algorithm}

\subsection{Collaborative Training of EMMs}
Multiple devices can collaborate in parallel to train $N$. The benefit of parallel collaboration is (1) we arrive at a converged/stable model much faster, and at a rate which scales with the number of collaborators, and (2) we increase the likelihood of observing rare yet benign occurrences, which reduces our false positive rate. To collaborate across multiple devices, we assume that the devices' hardware, kernel, and target application (being monitored) are of the same version. For example, all of the Samsung smart-fridges of the same model.  

Let $\mathbf{N}$ be a set of EMM models, where $\mathbf{N}^{(k)}_{ij}$ is the element $N_{ij}$ in the $k$-th model of $\mathbf{N}$. 
Since we assume that multiple devices statistically observe the same state sequences, then EMMs ($\mathbf{N}$) trained at separate locations can be merged into a single EMM. Since an EMM is a frequency matrix, we can combine the models by simply adding the matrices together by 
\begin{equation}\label{eq:simplecombine}
N^*=[N_{ij}^*]=[\sum_k \mathbf{N}^{(k)}_{ij}]
\end{equation}

It is critical that the combined model $N^*$ be trained on normal behaviors only. However, we cannot assume that all models have not been negatively affected. We propose two security mechanisms to protect $N^*$: model-attestation and abnormality-filtration.

\begin{description}
	
	\item[Abnormality-filtration] is used to combine a set of models $\mathbf{N}$ into a single model $N^*$, in manner which is more robust to noise and adversarial attacks than (\ref{eq:simplecombine}). The approach is to filter out transitions found in $N^*$ if the majority of models in $\mathbf{N}$ have not observed the same transition. To produce $N^*$ from $\mathbf{N}$ in our framework, Algorithm \ref{alg:combine} is performed, where $p_a$ the minimum percent of devices which must observe a transition in order for it to be included into the combined model. After forming the MC model $M^*$ from $N^*$ using (\ref{eq:emm}), an agent can attest that $M^*$ is a verified model via model-attestation.
	
	\item[Model-attestation] is used to determine whether a trusted model $N^{(i)}$ is similar to a given model $N^{(j)}$. If $N^{(j)}$ is similar, than it is considered to be a verified model with respect to $N^{(i)}$. To determine the similarity, we measure the linear distance between the EMMs, defined as
	\begin{equation}\label{eq:attest}
	d(N^{(i)},N^{(j)}) = \frac{\sum_{k=1}^{\text{dim}(M)}\sum_{l=1}^{\text{dim}(M)} |M^{(i)}_{kl} - M^{(j)}_{kl}| }{\text{dim}(M)^2}
	\end{equation}
	where dim$(M)$ is the length of $M$'s dimensions, and $M$ is the Markov chain obtained from $N$. A local device $i$ can attest that model $N^{(j)}$ is a self-similar model, if $d(N^{(i)},N^{(j)})<\alpha$, where $\alpha$ is a parameter of our framework (see Section \ref{sec:ciota}).
	
\end{description}

\begin{algorithm}
	\caption{The algorithm for combining a set of EMMs.}
	\label{alg:combine}
	\begin{algorithmic}[1]
		\Function{Combine}{$\mathbf{N}$, $p_a$}
		\State $N \leftarrow$ empty\_EMM$()$ \Comment{initialize empty freq. matrix}
		\For{$n_{ij} \in N$}
		\State $C \gets 0$ \Comment{init the counter}
		\For{$k \in 1:|\mathbf{N}|$}
		\State $n_{ij} \gets n_{ij} + \mathbf{N}^{(k)}_{ij}$
		\If{$\mathbf{N}^{(k)}_{ij} > 0$}
		\State $C++$
		\EndIf
		\EndFor
		\If{$\frac{C}{|\mathbf{N}|} \leq p_a$} \Comment{not enough devices have observed $ij$}
		\State $n_{ij} \gets 0$
		\EndIf
		\EndFor
		\State{\Return $N$}
		\EndFunction
	\end{algorithmic}	
\end{algorithm}

\section{The Framework}\label{sec:ciota}
In this section, we present the proposed framework and protocol. The framework enables distributed devices to safely and autonomously train anomaly detection models (Section \ref{sec:anom}), by utilizing concepts from the block chain protocol. 

First we will provide an overview and intuition of the framework (\ref{subsec:overview}). Then we will present the terminology which we use to describe the blockchain protocol (\ref{subsec:terms}). Finally, we will present the protocol and discuss its operation (\ref{subsec:protocol}). Later in Section \ref{sec:discussion}, we will discuss the various challenges and design considerations.

\subsection{Overview}\label{subsec:overview}
The purpose of the framework is to provide a means for IoTs to perform anomaly detection on themselves, and to autonomously collaborate to find the anomaly detection model. 
For example, a company may want gradually deploy thousands or millions of IoT devices. Each of the devices have an application, such as a web server (so that the user can interface and configure the device). The application may have un/known vulnerabilities which can be exploited by an attacker to accomplish some nefarious task. To detect threats affecting the devices, the company installs an agent on each device, and has the agent monitor the application.\footnote{An agent can cover multiple applications on a single device by maintaining separate models and blockchains. For simplicity, we will focus on protecting a single application.} 

The job of an agent is to (1) learn the normal behavior of the application, and (2) report abnormal activity in the application (Algorithm \ref{alg:anomDetect} on the local model $M^{(\ell)}$), and (3) report abnormal agents compromised or infected with malware. 

Each agent then collaborates with the other agents by trying to figure out how to safely combine everybody's local models into a single global model $M^{(g_1)}$. Once the agents agree upon a global model, each device will replace their $M^{(\ell)}$ with $M^{(\mathit{g}_1)}$. The agents continue to update their $M^{(\ell)}$ and collaborate on creating $M^{(\mathit{g}_2)}$. This collaboration cycle repeats indefinitely.

The benefit of collaboration is:
\begin{enumerate}
	\item An agent who has accidentally trained his $M^{(\ell)}$ on malicious behaviors will now detect them as malicious.
	\item The agents will benefit from the vast experience of all the devices together, and accurately classify rare benign events.
	\item The agents will be able to identify rouge agents by detecting corrupt \textit{partial-blocks} which fail model-attestation.  
\end{enumerate}

A critical part of the collaboration process is filtering out rare benign behaviors from possible malicious behaviors. 
The difference between the two is that we expect to see rare benign behaviors among more devices than malicious behaviors, especially at the outset of an attack (e.g., the propagation of a worm). This is relative to the parameter $p_a$: we expect at least $p_a\%$ of the agents to experience the rare-benign events, and less than $p_a\%$ to be infected. Note that after $M^{(g)}$ converges the malicious behaviors are detected, and $M^{(\ell)}$ is not updated with detected malicious behaviors (detailed in the protocol later on).

Since agents do not update their $M^{(\ell)}$ when an anomaly is detected, we expect each of the local models to remain pure.
However, there are cases where an $M^{(\ell)}$ can be corrupted. For example, when an agent launches after a malicious behavior begins, but before $M^{(\mathit{g}_1)}$ has been created.  
To protect the integrity of the next global model, when an agent which receives a set of local models (under collaboration to become the next global model), an agent will$\ldots$
\begin{enumerate}
	\item \textbf{[\textit{trust}]} $\ldots$consider only sets which contain authenticated models from different agents.
	\item \textbf{[\textit{filter}]} $\ldots$combine the set into a potential $M^{(\mathit{g})}$, and remove behaviors from $M^{(\mathit{g})}$ which have not been reported by the majority of agents (\textit{abnormality-filtration}).
	\item \textbf{[\textit{attest}]} $\ldots$accept the set of models as a potential $M^{(\mathit{g})}$, if it does not conflict with the agent's current local model (\textit{model-attestation}).
	\item \textbf{[\textit{inform}]} $\ldots$share the accepted set of models (including his own) with other agents, while reporting abnormal application behaviors and problematic agents (rejected \textit{partial-blocks}).
\end{enumerate}

The following analogy provides some intuition for how the agents create $M^{(\mathit{g})}$: 
\begin{tcolorbox}[breakable,title=\textit{Analogy}]
	A group of painters (agents) are looking at the same colored object (target application), and they are working together to select a single colored paint ($M^{(\mathit{g})}$) to describe it. Each painter produces a bucket of paint ($M^{(\ell)}$) based on their perception of the object's color. The painters then share their paint buckets with their neighbors, who mix the received paints together, while filtering out imperfections (\textit{abnormality-filtration}), but only if they feel the resulting color will still resemble the colored object (\textit{model-attestation}). The painters continue to adjust the paints, and after a set number of iterations of sharing, each painter pours some of the paint onto his/her pallate, and uses it to paint (perform anomaly detection). Then, the cycle repeats as the painters continue to adjust, filter, and share the paints in hopes of perfecting the color.
\end{tcolorbox}

To enable this autonomous trusted distributed collaboration, we use a \textit{blockchain}. In the following sections, we will detail how blockchain is used for this purpose.

\subsection{Terminology \& Notation}\label{subsec:terms}
In the framework, a blockchain is a linked list of sequential blocks, where each block contains a set of records (EMM models) acquired from different IoT devices of the same type (see Fig. \ref{fig:block}). Each device maintains a copy of the latest chain, and collaborates on the next block.

We will now list the terminology and notations necessary to explain the framework in detail:

\begin{description}
	\item[Model] A Markov chain anomaly detection model denoted $M$, where $N$ denotes the model in its EMM frequency matrix form. The model supports (1) the calculation of a distance between models, and (2) combining (merging) several models of the same type together. In this version, we use an EMM. We denote a model which is currently deployed on the local device as $N^{(\ell)}$.
	
	\item[Combined Model] A model created by merging a set of models together. The combined model only contains elements (transitions) which are present in at least $p_a$ percent of the models (see \textit{abnormality-filtration} in Algorithm \ref{alg:combine}). 
	
	\item[Verified Model] Let $d(M^{(i)},M^{(j)})$ be the distance between models $M^{(i)}$ and $M^{(j)}$. A model $M^{(i)}$ is said to be verified by a device if $d(M^{(i)},M^{(\ell)}) < \alpha$, where $\alpha$ is a parameter given by the user (see \textit{model-attestation} in (\ref{eq:attest})).
	
	\item[Record] A record is an entry in a block. A record consists of the model $N^{(i)}$ from device $i$, and a digital signature $S_{k_i}(\texttt{m}, n, \textbf{N})$, where $k_i$ is device $i$'s private key, \texttt{m} is the blockchain's meta-data (hash of previous block, target application, version$\ldots$), $\textbf{N}$ is the set of models from the start of the current block up to and including $N^{(i)}$, and $n$ is a counter which is incremented with each new block. The purpose of $n$ is to track the length of the chain and to prevent replay attacks. A record is \textit{valid} if the format is correct and the signature can be verified using device the corresponding device's public key. 
	
	\item[Block] A list of exactly $L$ records from different devices and some metadata. Each record is verified by the agent's digital signatures, where each agent's signature covers its model, all preceding models, the current block number ($n$), and its metadata (e.g., the agents' IP addresses).  We denote the $i$-th record in a block as $r_i$. The models in a block, when combined, represent a collaborative model $M^{(\mathit{g})}$ which can be used to replace a local model $M^{(\ell)}$. A block is \textit{valid} if the format is correct and contains valid records.
	
	\item[Partial Block] The same as a block, but it is less than $L$ entries long. The combined models in a \textit{partial-block} represent a proposed collaborative model $M^{(g)}$ in progress. An agent contributes (add its own model) to a \textit{partial-block} only if (1) the \textit{partial-block} is valid, (2) the agent does not already have a record there, and (3) the combined model, using the enclosed models, form a verified model (with respect to the agent's local model $N^{(\ell)}$). The length of a \textit{partial-block}, in the perspective of agent $i$, is the number of records in the \textit{partial-block} minus $i$'s if it exists.
	
	\item[Chain] A series of blocks (blockchain), where each block contains a hash of the previous block in \texttt{m}, and where the counter $n$ is the index of the block in the chain ($n=1$ for the first block, etc.) A chain contains the current collaboration for the next $M^{(\mathit{g})}$ (\textit{partial-block}), the current model with consensus (the last block in the chain), and an optional history of collaboration used for analytical purposes (all other blocks). The length of a chain is defined as the total number of full blocks in that chain. Finally, a chain may have at most one \textit{partial block} appended to the end of the chain. We denote the $i$-th block in a chain as $B_i$.
	
	\item[Agent] A program that runs on an IoT device which is responsible for (1) training and executing the local model $N^{(\ell)}$, (2) downloading more advanced broadcasted chains to replace $N^{(\ell)}$ and the locally stored chain, (3) periodically broadcasting the locally stored chain, with the agent's latest $N^{(\ell)}$ as a record in the \textit{partial block}, and (4) reporting any anomalous behaviors/blocks.
	
\end{description}

\begin{figure}[!t]
	\centering
	\includegraphics[width=.84\columnwidth]{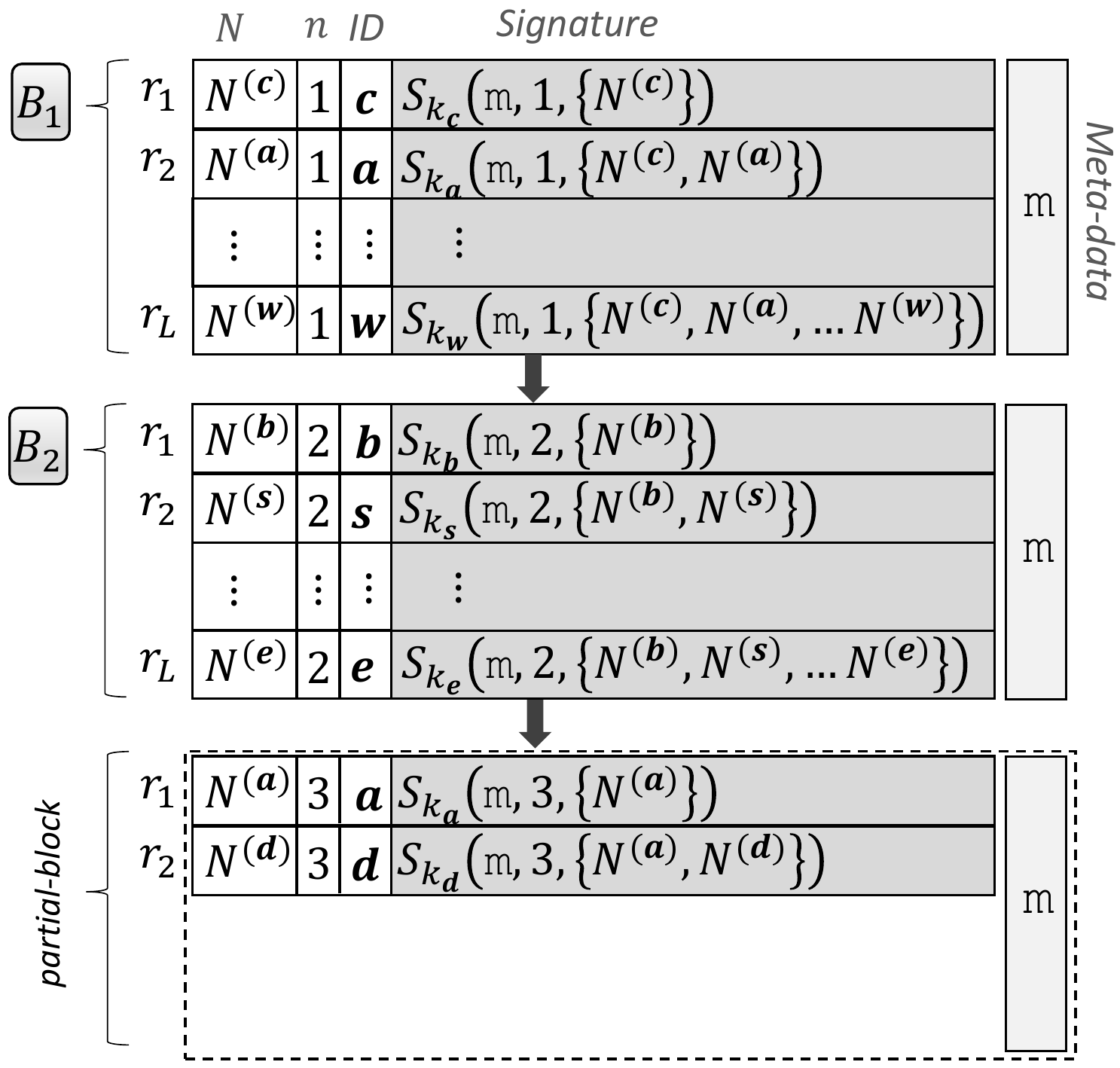} 
	\caption{An example of a chain with two blocks and a partial block, where the device IDs are $\{\mathbf{a},\mathbf{b},\mathbf{c}\ldots\}$.}
	\label{fig:block}
\end{figure}

\subsection{The Blockchain Protocol}\label{subsec:protocol}
By using a \textit{block-chain}, agents are able to collaborate autonomously in manner which is robust to adversarial attacks. Every agent maintains a local copy of the `best' chain. 

Closed blocks in the chain represent past completed global models, where the last completed block in the chain contains the most recently accepted model $M^{(\mathit{g})}$. The next global model is collaborated via a \textit{partial-block} appended to the chain. A \textit{partial block} only grows if agents can verify that it contains a safe model that captures the training distribution (the target app's behaviors). This is accomplished through trust propagation: agents (1) broadcast their \textit{partial block} to other agents, (2) replace their local \textit{partial-block} with received ones if they are both longer and similar to $N^{(\ell)}$ (same distribution check via \textit{model attestation}), and (3) reject and report \textit{partial-blocks} that are significantly different than $N^{(\ell)}$.


The blockchain protocol is as follows (illustrated in the flow-chart of Fig. \ref{LLD}):
\begin{tcolorbox}[breakable,title=\textit{Blockchain Protocol}]
	\singlespacing \vspace{-1.5em}
	\begin{enumerate}[leftmargin=*,label=\Alph*.]
		\item \textbf{Initialize.} An agent starts with an empty chain (an empty \textit{partial-block} with no preceding blocks) stored locally on its device, and initializes an empty local model $N^{(\ell)}$. 
		\item \textbf{Gather Intelligence (Monitor).} The agent (1) monitors the target application, (2) updates $N^{(\ell)}$ incrementally, and (3) reports anomalies if $T_{grace}$ has passed (Algorithm \ref{alg:monitor}). 
		\item \textbf{Share Intelligence.} Every $T$ seconds:
		\begin{enumerate}[label*=\arabic*.]
			\item \label{step:add_self} The agent adds its own local model $N^{(\ell)}$ to the \textit{partial-block} as a record, if $M^{(\ell)}$ is stable (passed $T_{grace}$), and does not yet exist in the \textit{partial-block}. 
			\item \label{step:broadcast} The agent shares its \textit{block-chain} (\textit{partial-block} and all preceding blocks) with $b$ other agents in a random order.\footnote{The agent only needs to broadcast the chain to a few `neighboring' agents, similar to how Etherium and Bitcoin work.} 
		\end{enumerate} 
		
		\item \textbf{Receive Intelligence.} When an agent receives a \textit{block-chain}:\footnote{To avoid DoS attacks, an agent will at most process $b$ chains once every $T$ seconds.}
		\begin{enumerate}[label*=\arabic*.]
			\item \textbf{If} the chain is shorter than the local chain: \textit{then} the agent discards the received chain.
			\item \textbf{If} the chain is longer than the local chain: \textit{then} the agent checks...
			\begin{enumerate}[label*=\arabic*.]
				\item \textbf{If} the last block is a valid block: \textit{then} the received chain replaces the local chain, and the models $\mathbf{N}$ in the last block are combined (\textit{abnormality-filtration}) to form $N^{(\mathit{g})}$ which replaces $N^{(\ell)}$.\footnote{The agent does not perform \textit{model-attestation} on a valid block.}\footnote{Option: Agents update $T$ to be a factor of the number of closed blocks in the local chain. Since $M^{(g)}$ converges over time, it is safer to prolong changes to the next version, increasing the response time when an attack on the blockchain is detected. See Section \ref{subec:adversarial} for details.}
				\item \textbf{Else}: the agent discards the received chain.
			\end{enumerate}
			\item \textbf{If} the chain has the same length as the local chain: \textit{then} the agent checks...
			\begin{enumerate}[label*=\arabic*.]
				\item \label{step:pb_accept}\textbf{If} (1) the received chain's \textit{partial-block} is longer than the local chain's \textit{partial-block} (excluding his own record from both), (2) the received \textit{partial-block} is valid, and (3) the models $\mathbf{N}$ in the \textit{partial-block} form a combined model (\textit{abnormality-filtration}) which the agent can attest is a verified model (\textit{model-attestation}): \textit{then} the received chain replaces the local chain.
				\item \label{step:pb_message} \textbf{Else If} (1) the chain's \textit{partial-block} has the same length as local chain's \textit{partial-block} (excluding his own record), (2) the two \textit{partial-blocks} have different agent IDs, (3) the \textit{partial-block} is valid, and (4) this is the $k$-th received chain of equal length whose \textit{partial-block} was that was not used: \textit{then} send the local chain to the agent(s) in received \textit{partial-block} who do not appear in the local \textit{partial-block}.\footnote{The received \textit{partial-block} has the IP addresses of the target agents.}  
				\item \label{step:pb_reject} \textbf{Else}: (1) the agent discards the received chain, and (2) \textbf{If} in steps \ref{step:pb_accept} or \ref{step:pb_message} the \textit{partial-block} failed the validity check or failed the \textit{model-attestation} to a significant degree, \textit{then} report the block and sending agent.\footnote{Alternative version: If the last block is valid yet different than the local chain's, then merge that block's combined model into $N^{(\ell)}$. This helps form a more general $N^{(g)}$ without communities. A limitation must be placed on the number of merges per $T$ seconds.}
			\end{enumerate} 
		\end{enumerate}
	\end{enumerate}
\end{tcolorbox}

\begin{figure}[!t]
	\centering
	\includegraphics[width=\textwidth]{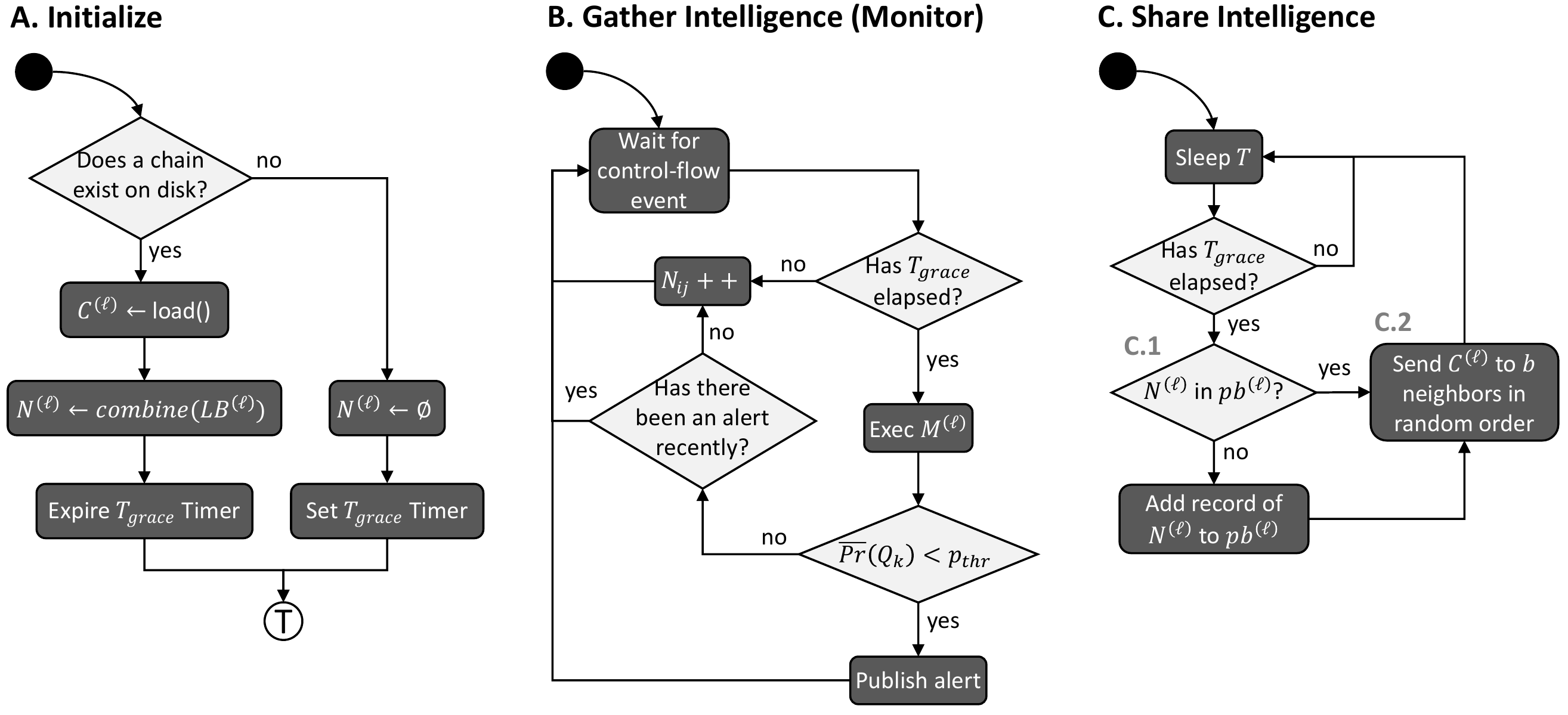}
	\includegraphics[width=\textwidth]{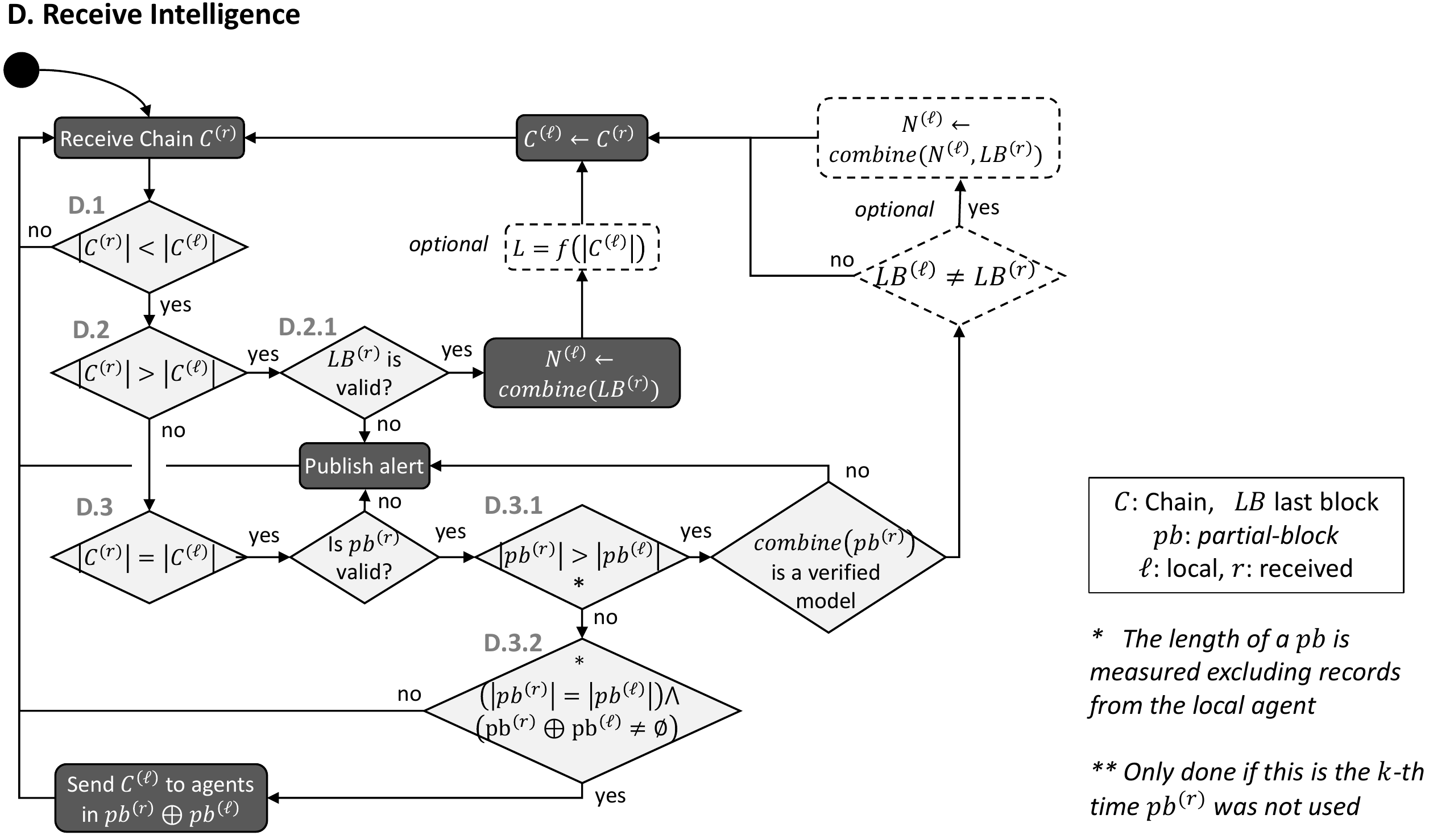} 
	\caption{A flow-chart of the blockchain protocol.}
	\label{LLD}
\end{figure}
\subsection{Proof of Cumulative Majority}
In blockchains, there is often some form of effort which deters an attacker from making false records. In systems like Etherium, it's the effort of solving a crypto challenge. This type of challenge is necessary in systems like Etherium, because there is a base assumption that all participants are untrusted from the start. In contrast, our system assumes that the majority participants (agents) are on uncompromised devices at the start, because they are deployed by the manufacture. This is a common assumption for IDS systems.

Therefore, this blockchain uses ``proof of cumulative  majority'' to deter attacks. The cumulative majority refers to the distributed consensus, or significant mass, achieved by accumulating $L$ the participants' signatures on a set of models to be combined as the next global model.

Concretely, an agent only replaces its local \textit{partial-block} with a received \textit{partial-block} if it is similar to the behaviors is has seen locally (\textit{model-attestation}). Therefore, a \textit{partial-block} of length $L$ can only exist if $L$ agents can attest that the model is similar to their own model/observations (i.e., there are $L$ compromised agents within $T$ seconds). Since $L$ is very large in practice (10k-100k), and \textit{partial-blocks} are indiscriminately shared and propagated: (1) a closed block has majority trust on it, and (2) is unlikely to be malicious due to the attacker's significant challenge.

The attacker's challenge/effort in this blockchain is to compromise a significant number of devices before $T$ seconds pass. Otherwise, the attack is reported in step \ref{step:pb_reject} of the protocol. At which point, the attack is discovered and (1) the affected devices' keys can be invalidated, and (2) the devices can be cleaned and patched.

With this in mind, we can see how the proposed blockchain system achieves its objective as an IDS. When a device is compromised, either (1) the agent will detect an abnormal behavior and report it to the SOC, or (2) the model will be corrupted/tainted by the latent behavior. In the latter case, if the compromised device publishes its model, it will be rejected by the other devices, because the tainted partial block (PB) will no longer be self-similar to the other devices’ models in the \textit{model-attestation} step. The other agents will then report rejected blocks, for example, to a Security Operations Center (SOC), and it will be clear who the infected device is (identified with problematic model's key from the reported PBs). The SOC can then invalidate that device's key and investigate the intrusion. Therefore, if the agent is compromised then the tainted model will be detected by the community, and if the model is corrupt (contains abnormal behaviors) then the agent will detect the intrusion when it replaces the local model with the next global model.

\subsection{Model Conflicts in Partial Blocks}\label{subsec:conflicts}
A concern might be that the agents will disagree on the models in the \textit{partial-block} and not reach a consensus. However, all agents monitor the same application running on the same type of hardware. Therefore, their models are very similar to one another. This is intuitive because each agent's training data follows the same distribution, and the Markov chain captures the probabilities of PC transitions. 

Since the models are trained on the same distribution, any model formed by combining a subset of all agents' models will also be similar all agents' models. More formally, we observe that
\begin{equation}
	d\left(combine\left(N_{i}^{(\ell)}\right),n_{j}^{(\ell)}\right)<\alpha \hspace{1em} \forall i,j : n_{j}^{(\ell)}\in N^{(\ell)},N_{i}^{(\ell)}\in \mathbf{N}^{(\ell)}
\end{equation}
where $\mathbf{N}^{(\ell)}$ is the set of all agents' models, and $d$ is the average parameter distance defined in (\ref{eq:attest}). This holds true since all agents are sampling from the same distribution (hardware and software). In our experiments, we were able to set $\alpha$ to a low value because the agents' benign models were consistently very similar (Section \ref{sec:eval}).
Therefore, it is highly unlikely that the \textit{partial-block} will be in conflict given a reasonable $\alpha$. 

\subsection{Deadlock Prevention}\label{subsec:deadlocks}
As mentioned in the protocol, agents should only message a few other agents in step \ref{step:broadcast} to minimize traffic overhead. However, a deadlock can occur if (1) connectivity between agents is incomplete (some agent's cannot directly message other agents), (2) all agents have their neighbor's records in their \textit{partial-block}, and (3) all \textit{partial-blocks} have the same length. Although it is very rare for this to occur (one in a million depending on the connectivity), step \ref{step:pb_message} prevents any deadlocks that may happen.

The following is the formal proof that our revised system will not have any deadlocks in reaching a \textit{partial-block} of length $L$. 

Let the undirected graph $G=(E,A)$ represent the agent's connectivity, where $i\in A$ is the set of agent IDs. Let $pb_i$  be the \textit{partial-block} of agent $i$ such that $pb_i \subseteq A$. We denote the set of neighbors which are directly connected to agent $i$ as $\Gamma_i$. Finally, we refer to an epoch as an iteration where all agents have broadcasted a their $pb$ to their neighbors (every $T$ seconds).

In our proof, we assume that $G$ forms a single connected component. We also assume that $L=|A|$ because if a $pb$ reaches length $|A|$ then it will reach all possible $L$, where $L\leq |A|$. We also assume that all agents are drawing observations from the same distribution to train their models, and therefore will not have any issue during the $pb$ validation checks (Section \ref{subsec:conflicts}). 

A deadlock occurs if $\forall i\in A:D(i)$ where the predicate $D$ is defined as $D(i) : pb_{i}^{(t)}=pb_{i}^{(t+1)} \land |pb_{i}^{(t)}|<L$.

\begin{lemma}\label{lemma:have_own}
	If there was no update after an epoch (deadlock), then all agents have their own ID in their partial block. Formally,
	$\forall i\in A:D(i) \rightarrow \forall i \in A : pb_i \cap  \{ i \}= \{ i \}$.
\end{lemma}
\begin{proof}
 Let's assume that $\forall i \in A : D(i)$ and that $\exists i \in A : pb_i \cap \{i\}=\emptyset$. This cannot be because of step \ref{step:add_self} of the protocol: every agent $i$ adds record `$i$' to $pb_i$ if $pb_i \cap \{i\}=\emptyset$. Therefore, $\forall i \in A : D(i) \rightarrow \forall i \in A : pb_i \cap \{i\}=\{i\}$. 
\end{proof}

\begin{lemma}\label{lemma:have_eachother}
If there was no update after an epoch (deadlock), and agents $i$ and $j$ are neighbors, then both agents have IDs $i$ and $j$ in their partial blocks. Formally, $\forall i \in A : D(i)\rightarrow \forall (i,j) \in E : pb_i \cap \{i,j\}=\{i,j\}\}$.
\end{lemma}
\begin{proof}
If we prove $\forall i \in A : D(i) \rightarrow \forall (i,j) \in E : pb_i \cap \{i\} \cap pb_j = \{i\}$, then we have proven Lemma \ref{lemma:have_eachother} by symmetry: Let's assume $\forall i \in A : D(i)$ but $\exists (i,j) \in E : pb_j \cap \{i\} = \emptyset$. This could not be true because (1) agent $i$ shared its $pb$ with agent $j$ and vice versa (step \ref{step:broadcast} of the protocol), (2) $pb_i \cap \{i\}=\{i\}$ and $pb_j \cap \{j\}=\{j\}$ (Lemma \ref{lemma:have_own}, and (3) in all cases, agent $i$ would have replaced it's pb with $j$'s (or vice versa):

\textit{Case 1}: $|pb_i|=|pb_j|$ and $pb_i$ either has $j$ or not. If $pb_i$ has $j$ then agent $i$ should have replaced $pb_i$ with $pb_j$ because $|pb_i \oplus \{i\}|<|pb_j \oplus \{i\}|$ (step \ref{step:pb_accept} of the protocol). Similarly, if $pb_i$ doesn't have $j$ then agent $j$ would have replaced $pb_j$ with $pb_i$ because $|pb_j\oplus \{j\}|<|pb_i \oplus \{j\}|$.

\textit{Case 2}: If $|pb_i|<|pb_j|$ then agent $i$ would have replaced $pb_i$ with $pb_j$ because $|pb_i \oplus\{i\}|<|pb_j \oplus \{i\}|$ since $pb_j \cap \{i\}=\emptyset$.

\textit{Case 3}: If $|pb_i|>|pb_j|$ then there is only one case where $|pb_i \oplus \{i\}|=|pb_j \oplus \{i\}|$ resulting neither agent performing an update: $pb_i \cap \{j\}=\{j\}$ and $|pb_i|=|pb_j|-1$. However, because $pb_j \cap \{j\}=\{j\}$ (Lemma \ref{lemma:have_own}), $|pb_j \oplus \{j\}|<|pb_i \oplus \{j\}|$ so agent $j$ would have replaced $pb_j$ with $pb_i$
Therefore, we conclude that $\forall i \in A : D(i) \rightarrow \forall (i,j) \in E : pb_i \cap \{i\} \cap pb_j=\{i\}$, so Lemma \ref{lemma:have_eachother} holds true.
\end{proof}

\begin{lemma}\label{lemma:same_len}
Lemma 3. If there was no update (deadlock) then all partial blocks have the same length. Formally, $\forall i \in A : D(i) \rightarrow \forall ij \in A : |pb_i|=|pb_j|$
\end{lemma}
\begin{proof}
Let's assume that $\forall i \in A : D(i)$ but $\exists(i,j) \in E : |pb_i|<|pb_j|$. However, $pb_i \cap \{i,j\} \cap pb_j=\{i,j\}$ (Lemma \ref{lemma:have_eachother}). This means that $|pb_i \oplus \{i\}|<|pb_j \oplus \{i\}|$ so agent $i$ would have set $pb_i=pb_j$, and $\forall i \in A : D(i)$ would not hold true. Therefore, it must be that $|pb_i|=|pb_j|$.
\end{proof}
	
\begin{lemma}\label{lemma:bad_neighbors}
Lemma 4. If there was no update (deadlock), then there exist two neighbors with different partial blocks, of same length. Formally, $\forall i \in A : D(i) \rightarrow \exists(i,j) \in E : pb_i \neq pb_j$.
\end{lemma}
\begin{proof}
Let's assume $\forall i \in A : D(i)$ but $\forall (i,j) \in E : pb_i=pb_j$. According to Lemma \ref{lemma:have_own}, all agents have their own ID in their partial block. However, if all agents have the same partial block, then that means that $|pb_i|=L$ and $\forall i \in A : D(i)$ does not hold true. Therefore, it must be that $\exists (i,j) \in E : pb_i \neq pb_j$.
\end{proof}

\begin{theorem}\label{theorem:nodeadlock_eqL}
Given a set of agents $A$, the connectivity network $G$, and $L=|A|$, there will never be a deadlock. Formally, $L=|A|\rightarrow \nexists i \in A: D(i)$.
\end{theorem}
\begin{proof}
Let's assume that $L=|A|$ but $\forall i \in A : D(i)$ (there is a deadlock). This would mean that $\exists (i,j) \in E : pb_i \neq pb_j$ (Lemma \ref{lemma:bad_neighbors}). If so, it must be that there is at least one ID in $pb_j$ that is not in $pb_i$ (via Lemmas \ref{lemma:have_eachother} and \ref{lemma:same_len}). Let's say that one of these IDs is that of agent $k$. When agent $j$ shared $pb_j$ with agent $i$ (step \ref{step:broadcast} of the protocol), agent $i$ would have sent $pb_i$ directly to its non-neighbor $k$ (step \ref{step:pb_message} of the protocol). Since $|pb_k|=|pb_i|$ (Lemma \ref{lemma:same_len}), and $|pb_k \oplus \{k\}|<|pb_i \oplus \{k\}|$ because $pb_i$ does not have $k$, agent $k$ must have replaced $pb_k$ with $pb_i$. Therefore, it is impossible for $\forall i \in A : D(i)$ to hold true since in the next epoch, agent $k$ would have added itself to its partial block making $|pb_k|>|pb_i|$.  
\end{proof}
The continuation can be seen through Lemma \ref{lemma:same_len}: it must be that all other agents will grow their partial blocks to the same length as $pb_k$. Then, if there is another deadlock, the above process repeats until $\exists i \in A : |pb_i|=|A|$ and the block is closed.

\begin{corollary}\label{corr:deadlock_degree}
Corollary 1. Given a set of agents $A$, the connectivity network $G$, and $L \leq |A|$, there will never be a deadlock. Formally, $L \leq|A| \rightarrow \nexists i \in A : D(i)$.
\end{corollary}
\begin{proof}
The proof is trivial via Theorem \ref{theorem:nodeadlock_eqL} since there exists an agent that will reach a partial block length longer than $|A|$, and step \ref{step:add_self} of the protocol ensures that partial block grow in length by one at time.
\end{proof}

As a side note, if step \ref{step:pb_message} (direct messaging) is removed from the protocol, the system will reach a partial block length of the maximum degree plus one without any deadlocks:
\begin{theorem}\label{theorem:nodeadlock_leqL}
Theorem 2. Given a set of agents $A$, the connectivity network $G$, and $L=\Delta(G)+1$, there will never be a deadlock. Formally,  $L=\Delta(G)+1 \rightarrow \nexists i \in A : D(i)$.
\end{theorem}
\begin{proof}
Let's assume that agent $i$ has the maximum degree. According to Lemma \ref{lemma:have_eachother}, $pb_i$ must have all of its neighbor's IDs and $i$ before a deadlock can occur. Therefore, there can't be a deadlock because $|pb_i|=|\Delta(G)+1|=L$. 
\end{proof}

\subsection{Peer Discovery}\label{subsec:peerdisc}
To broadcast the latest chain, an agent must know the IP addresses of the receiving agents. It is important to note that an agent does not need to broadcast to all other agents. Instead, an agent broadcasts to $b$ other agents where $b$ is much smaller than the population size. In practice, $b$ can be in the order a tens or hundreds where there is a trade-off between the rate at which information is shared across the network (iterations of $T$) and the amount of work that is put into each broadcast. Regarding the discovery and selection of peers, we suggest that the Ethereum's p2p discovery protocol \cite{Discover45:online} be used and that an agent should periodically draw new peers at random.

\subsection{Maintaining Software Versions}\label{subsec:branching}
As time goes on, the target application may receive software updates during its software life-cycle. Although the app's new behavior will be accepted as normal (due to the majority consensus), there may be other devices where not yet updated or may never be updated. To ensure that these outdated devices aren't `forced' to use an incompatible model, we suggest that blockchain should support branching. In this approach, the chain forms a version tree were devices with newer versions can `fork' off to. To enable this the following additions are made to the protocol: (1) the respective software version must be stored in the metadata of each block, (2) multiple partial blocks of different version can be stored at the end of a chain, (3) if a partial block is completed but it has a different version than the current branch, then a separate chain is `forked' from that point, and (4) agents always follow the longest chain with their version. 

\section{System Evaluation}\label{sec:eval}
In this section, evaluate the proposed collaboration framework: the experiment testbed, parameters, results, and observations. A video demo of the framework is available online.\footnote{\textit{The short demo of the framework protecting 48 Pis running web servers can be found at \texttt{\url{https://youtu.be/T4t_SnTJV3w}}}}

\subsection{Experiment Setup}
Our experiments were composed of four aspects: the (1) test environment, (2) implementation, (2) target applications, and (3) attack scenarios. We will now discuss each of these aspects in detail.

\subsubsection{Test Environment}
We built a LAN which served as a simulation platform for emulating a distributed IoT environment (Fig. \ref{PiBoard}). 
This network involves 48 Raspberry Pis connected together through a single large switch.

In our environment, each Raspberry Pi was equipped with additional boards (sheilds) and sensors. For example, the PiCamera and Pibrella Board\footnote{\textit{Pibrella module can be found at \texttt{\url{www.pibrella.com}}}} which provides programmatic access to three LED lamps and simple 8-bit PC speaker. For each experiment, a target application (IoT software) was loaded and executed on all of the devices, along with an agent.

The source code for the agent can be found on GitHub.\footnote{The agent's code from the experiment can be found at \texttt{https://git.io/vAIvd}}

\subsubsection{Agent Implementation}
To implement Algorithm \ref{alg:monitor} (monitor), we implemented the agent using OS and CPU features. Specifically, we used the performance counters API and Core-sight (on ARM) and Last-Branch (on Intel). By using these libraries and features, we were able to track the application's control-flow in an asynchronous manner.

In our implementation, the kernel fills a large ring-buffer with observed jump and branch addresses. 
When the OS scheduler switches to the agent, the agent iterates over the new entries in the buffer and updates $M^{(\ell)}$ accordingly. To improve performance further, the agent was written entirely in C++. However, the code was not optimized to its full potential.

The underlying network protocol we used in our experiment was the UDP Multicast protocol, though in practice, the Bitcoin or Etherium P2P neighbor discovery algorithm should be used. The following lists the parameters used in all experiences, unless noted otherwise:
\begin{itemize}
	\item \textbf{\boldmath$T$ (Processing interval)}: one minute
	\item \textbf{\boldmath$L$ (Block size)}: $20$
	\item \textbf{\boldmath$p_a$ (Percent of reporting devices required to include a transition)}: $25\%$
	\item \textbf{\boldmath$\alpha$ (Verification distance)}: $0.05$
	\item \textbf{\boldmath$p_{thr}$ (Anomaly score threshold)}: $0.012$
	\item \textbf{\boldmath$k$ (Probability averaging window)}: $10,000$
	\item \textbf{Region size}: $256$ Bytes
\end{itemize}

\begin{figure*}	
	\centering
	\includegraphics[width=\textwidth]{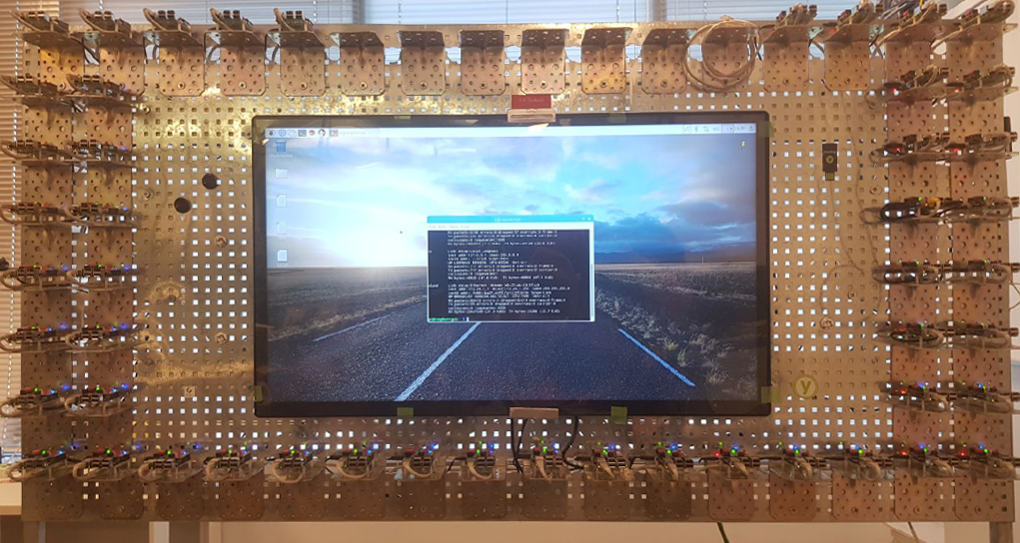}
	\caption{IoT simulation testbed consisting of 48 Raspberry Pis}
	\label{PiBoard}
\end{figure*}

\subsubsection{Target Applications}
Every application has a different control-flow, and reacts differently to environmental stimuli. Therefore, we evaluated the framework using several different target applications:

\begin{description}
	\item[Smart Light] Smart lights can perform custom functionalities programmed by the user. By evaluating the framework on a smart light, we are able to determine whether  each agent is able to learn its functionality, and how the propagation of these behaviors affect other agents.
	To implement the smart light's software, we combined several Open-Source projects \cite{mongoose, pibrellaGitHub, WiringPi}. The final application contained a vulnerable web-based interface for controlling the light's features.
	
	\item[Smart Camera] Smart cameras often consume a significant amount of resources to perform real-time image processing. By monitoring such an application, we are able to evaluate how well the framework performs in resource heavy applications. The application which we used monitors a video feed and sends an alert when it detects a movement. The alert is sernt to a control server and is accompanied with a short video or image of the event. A user interfaces with the camera via the server, and can either (1) change its configuration or (2) view the camera's current frame. We included with the final application a null dereference vulnerability in the communication process with the control server.
	
	\item[Router] Routers are widespread and provide Internet facing IPs (i.e., are not hidden behind a NAT). They are a good example of vulnerable IoTs which have been the target of many recent attacks (e.g., Mirai and the VPNFilter malware\footnote{\texttt{\url{https://www.symantec.com/blogs/threat-intelligence/vpnfilter-iot-malware}}}). By evaluating the framework on a router's software, we are able to consider how well our agent handles complex control-flows. Routers typically have a Linux kernel, and provide their functionality via several different applications. In our evaluation, we chose to target the Hostapd (Host access point daemon) applicaiton. Hostapd is a user space software access point capable of turning normal network interface cards into access points and authentication servers. We took version 2.6 of Hostapd which is vulnerable to a known replay attack.\footnote{The code is available at \texttt{\url{https://github.com/vanhoefm/krackattacks-scripts}}}
\end{description}

\subsubsection{Attack Scenarios}
To understand the framework's detection capabilities, we evaluated how well the agents can detect the exploitation of different vulnerabilities and the execution of malicious code:

\begin{description}
	\item[Buffer Overflow] When writing information into a buffer, without proper boundary checks, it is possible to write more data than the buffer's size. 
	When this occurs, the data overflows and overwrites the code and variables in memory.
	If executed correctly, a buffer overflow can be used to alter a programs code and alter the control-flow of the program. 
	This situation is dangerous because a crafted input data can contain machine instructions, thus causing the program to execute arbitrary code in the software's context \cite{deckard2005buffer}. In this scenario, we (1) exploit a buffer overflow vulnerability in the application, (2) covertly have the app behave like a bot, and (3) preserve the application's original behavior. The bot attempted to connect with a C\&C server once every minute.
	
	\item[Code-Reuse] Instead of injecting new code into the program's memory layout, a code-reuse attack \cite{prandini2012return,elreturn,bletsch2011jump} uses the existing code of the program to create a new logic, mostly by performing jumps to unusual places in the code. For example, jumping to the middle of functions or jumping multiple times to different instructions which perform the desired logic.
	These attacks were proved to be, in many cases, tuning complete \cite{tran2011expressiveness}. This means that an attacker can potentially cause a typical program to execute any desired logic.
	A common approach is called ``return-to-libc'' \cite{elreturn} which reuses code in the libc library to execute the desired code. More advanced approaches are to use the ROP (Return-oriented programming \cite{prandini2012return}) and JOP (Jump-oriented programming \cite{bletsch2011jump}) techniques.
	In this scenario, we attack perform a code-reuse attack on the target application in order to get the app to send sensitive data to a remote server.
	
	\item[Replay Attack (Key Reinstallation Attack)] The Key Reinstallation Attack is a type of replay attack in which one or more protocol's messages are sent again in a different, unexpected, point of the protocol. The Key Reinstallation Attack tries to leak information about encrypted traffic by changing the application state in the middle of the encryption process. Unlike the previous two attacks, this attack does not execute new arbitrary logic within the application's memory space, but rather abuses the control-flow to reveal encryption secrets which can be used to decrypt a user's traffic off-site. 
	
\end{description}

\subsubsection{The Experiments}
To evaluate the framework's anomaly detection capabilities with different applications and attacks, we used several different experiment setups summarized in Table  \ref{ExperimentsSummery}. We will refer to these experiments using their short-form notation from the table.

Unless stated differently, for every experiment, the target application and a local agent were launched on 48 Raspberry Pis simultaneously. After two hours, we paused the training and began to record the performance for another two hours. Finally, at the start of the fifth hour, the specified attack was executed. Although it was not part of the protocol, we paused the training in order to observe the performance of a collaborative model which has been trained for exactly two hours. It is critical that the target application would not remain dormant, but rather, is exposed to normal interactions like an IoT device. Therefore, to successfully simulate a real environment, during all of our experiments we legitimately interacted with the target application manually, using random fuzzing, and previously recorded data on the application's input channels. For example, we used a prerecorded video stream in the experiments involving the smart camera.

\begin{table}[!t]
	\begin{center}
		\caption{Summary of Experiment Setups}
		\begin{tabular}{c}
			\includegraphics[width=.6\columnwidth]{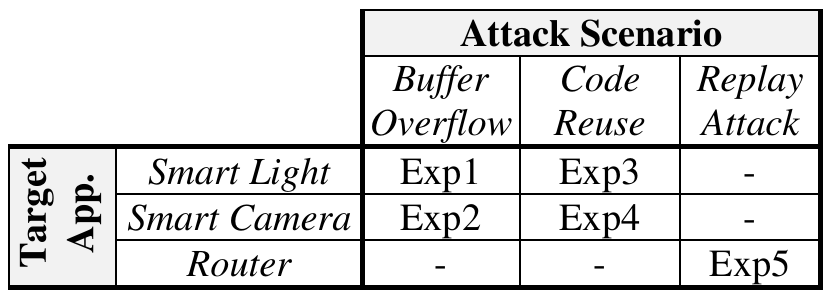}
		\end{tabular}
		\label{ExperimentsSummery}
	\end{center}
\end{table}

\subsection{Experiment Results}
The contributions of this paper are (1) a method for detecting abnormal control-flows (2) efficiently, and (3) a method for performing collaborative training (4) in the presence of an adversary. We will now present our results accordingly.

\subsubsection{Anomaly Detection}\label{subsubsec:anom}
We will now evaluate the use of EMMs over regions of an application's memory space as a method for anomaly detection, on a \textit{single} device.

The code injection attacks (buffer-overflow and code-reuse) were detected entirely with no false positives. Fig. \ref{EMM_ON_REPLAY:a} plots the EMM probability scores $\overline{Pr}(Q_{k})$ for Exp2. 
The Key Reinstallation Attack (Exp5) was more difficult to detect (Fig. \ref{EMM_ON_REPLAY:b}). This is because the attack does not inject own code, and the impact on the control-flow is very brief (a single step in the protocol). However, the attack still influences the probability scores, and we are able to detect the attack when $k$ is increased. Furthermore, when the train time is increased, the performance increases as well. This is evident in the collaborative training setting where two hours of training on 48 devices is equivalent to two days of training. In this case the EMM model yields perfect detection with no false positives. 

In summary, given enough train time, our proposed anomaly detection method is capable of detecting arbitrary code injection attacks and other kinds of exploits (such as protocol exploits). 

\begin{figure}[p]
	\centering
	\includegraphics[width=.8\columnwidth]{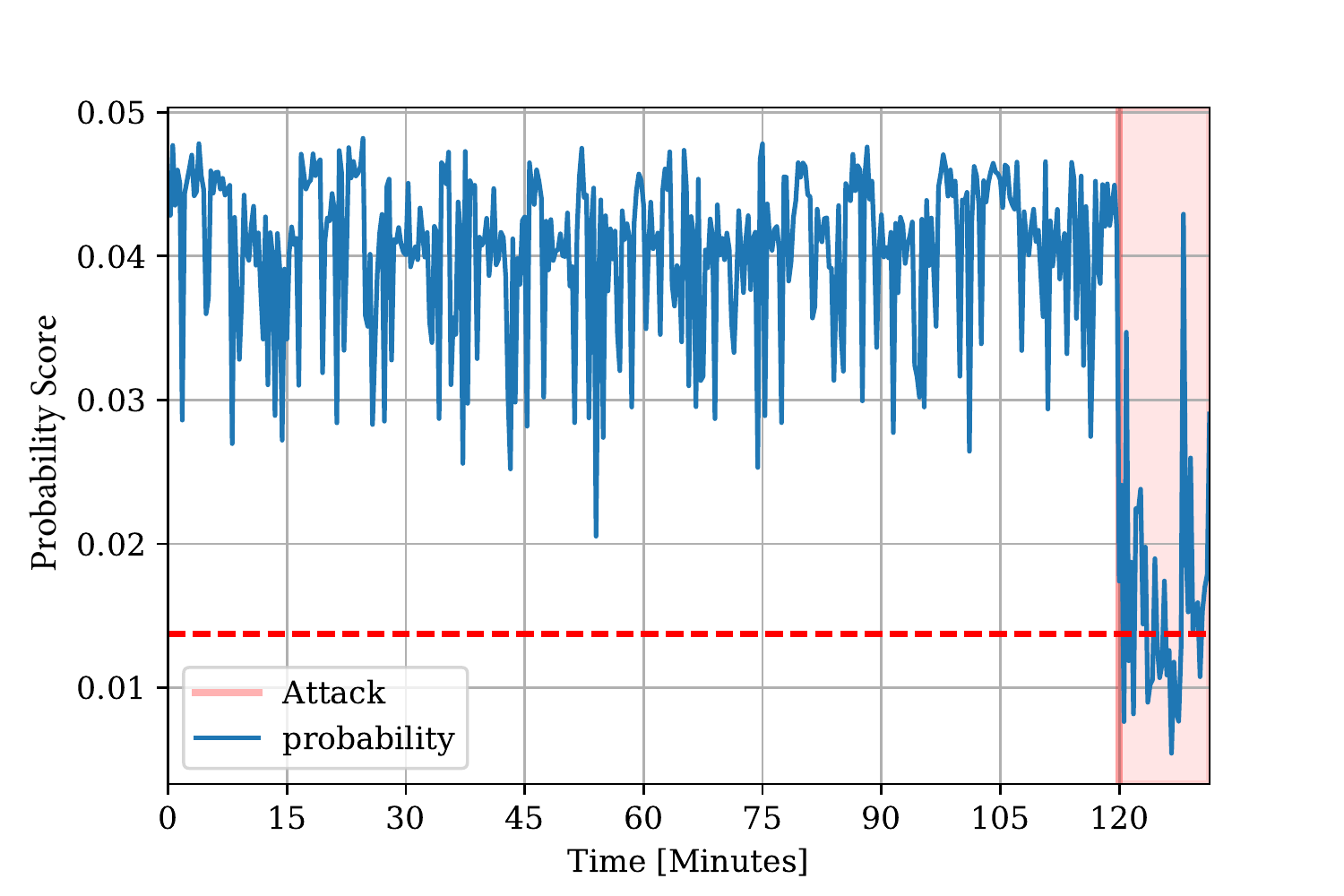}
	\label{EMM_ON_REPLAY:a}

	\caption{The probability scores of $M^{(\ell)}$ from Exp2 after two hours of training, where the red area marks the attack period.}
	\vspace{.3cm}
	\centering
	\includegraphics[width=.8\columnwidth]{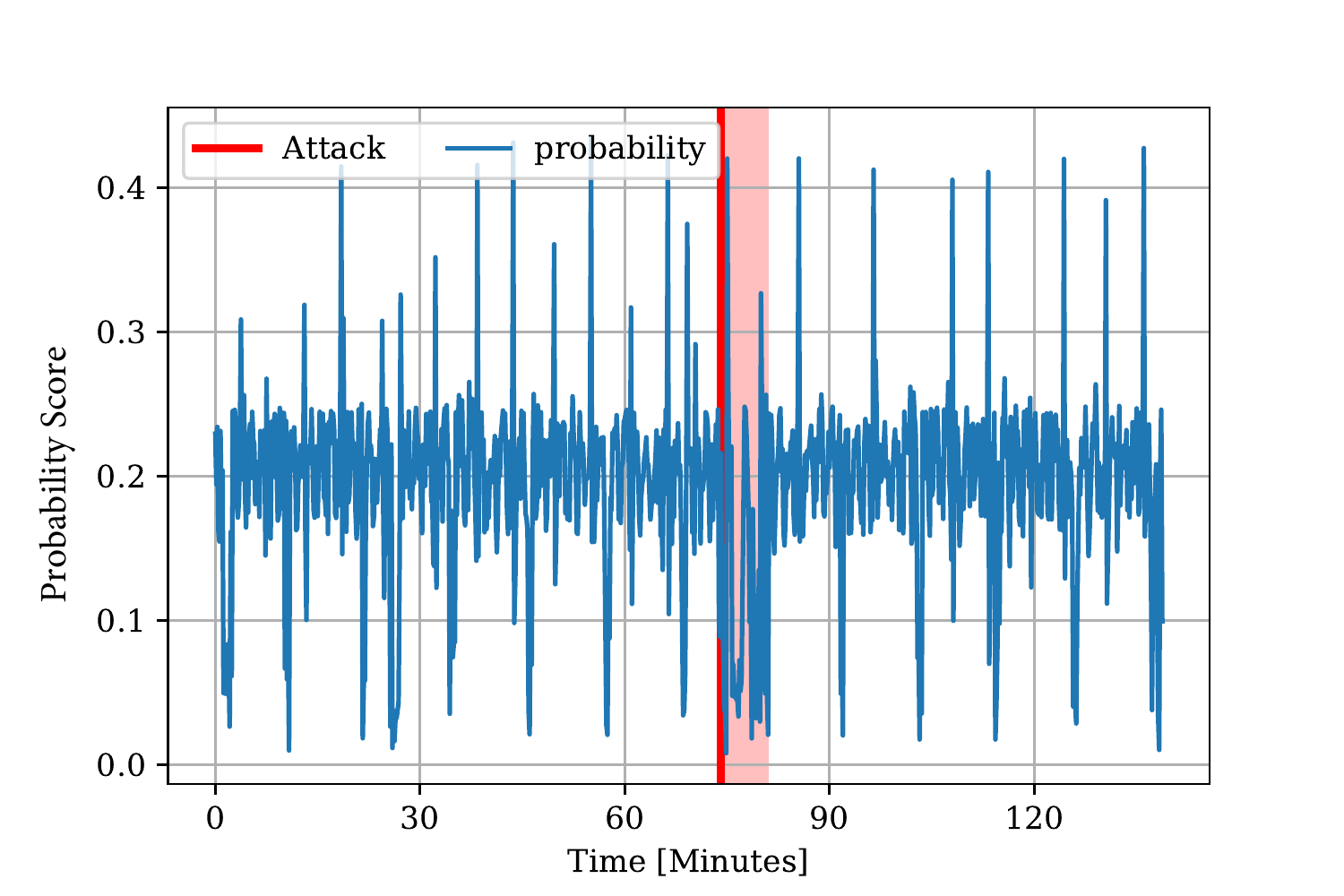}
	\label{EMM_ON_REPLAY:b}

	\caption{The probability scores of $M^{(\ell)}$ from Exp5 after two hours of training, where the red area marks the attack period.}
\end{figure}

\subsubsection{Collaboration Training}
In section \ref{subsubsec:anom}, we showed how EMMs can detect a variety of attacks on IoT devices, given enough train time. However, an anomaly detection model is vulnerable during its \textit{initial} train time ($T_{grace}$). Furthermore, a single device may not experience all possible behaviors in the alloted time. In contrast, collaborative training, using multiple IoT devices, can produce a model a shorter period of time which performs better.

\begin{description}
	\item[Model Performance] By performing collaborative learning, the final model contains the collective experiences from many different devices. As a result, each device can better differentiate between rare-benign behaviors and malicious behaviors. Fig. \ref{Collaborative_Training_Exp:a} shows that the same amount of train time distributed over 48 devices produces a model which can detect an attack sooner than when simply performing all of the train time on a single device. The reason for this is the distributed model captures a more diverse set of behaviors, which helps it differentiate better between malicious and benign.   
	\item[Model Train Time] Fig. \ref{Collaborative_Training_Exp:b} shows that several models trained in parallel can produce a stronger model than a single model (Fig. \ref{Collaborative_Training_Exp:a}) in the same amount of time. Thus, we see that $M^{(g)}$ converges at a rate which is inverse to size of the network. As a result, a large IoT deployment will obtain a strong model quickly, and is much less likely to fall victim to an adversarial attack.
\end{description}

\begin{table}[h]
	\begin{center}
		\caption{False Positive Rates with Collaborative Learning: All Attack Scenarios}
		\begin{tabular}{c}
			\includegraphics[width=.6\columnwidth]{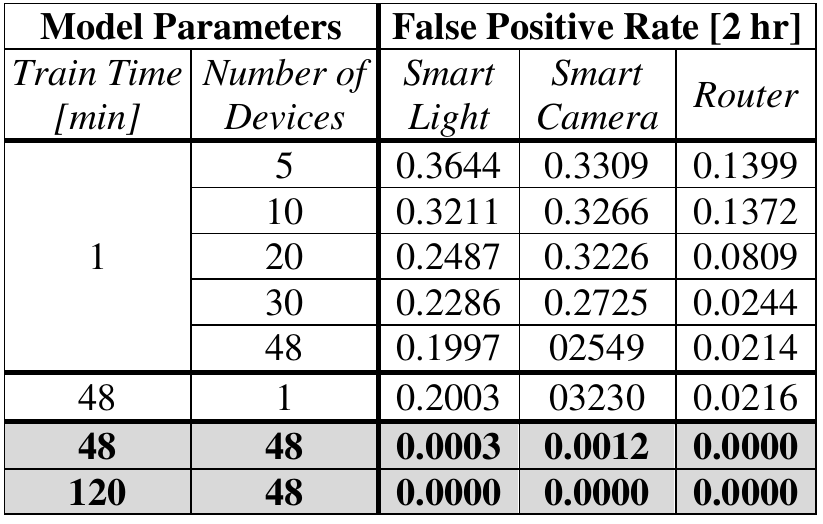}
		\end{tabular}
		\label{Collaborative_Training_ExpSummery}
	\end{center}
\end{table}

\begin{figure}[p]
	\centering
	\includegraphics[width=.8\columnwidth]{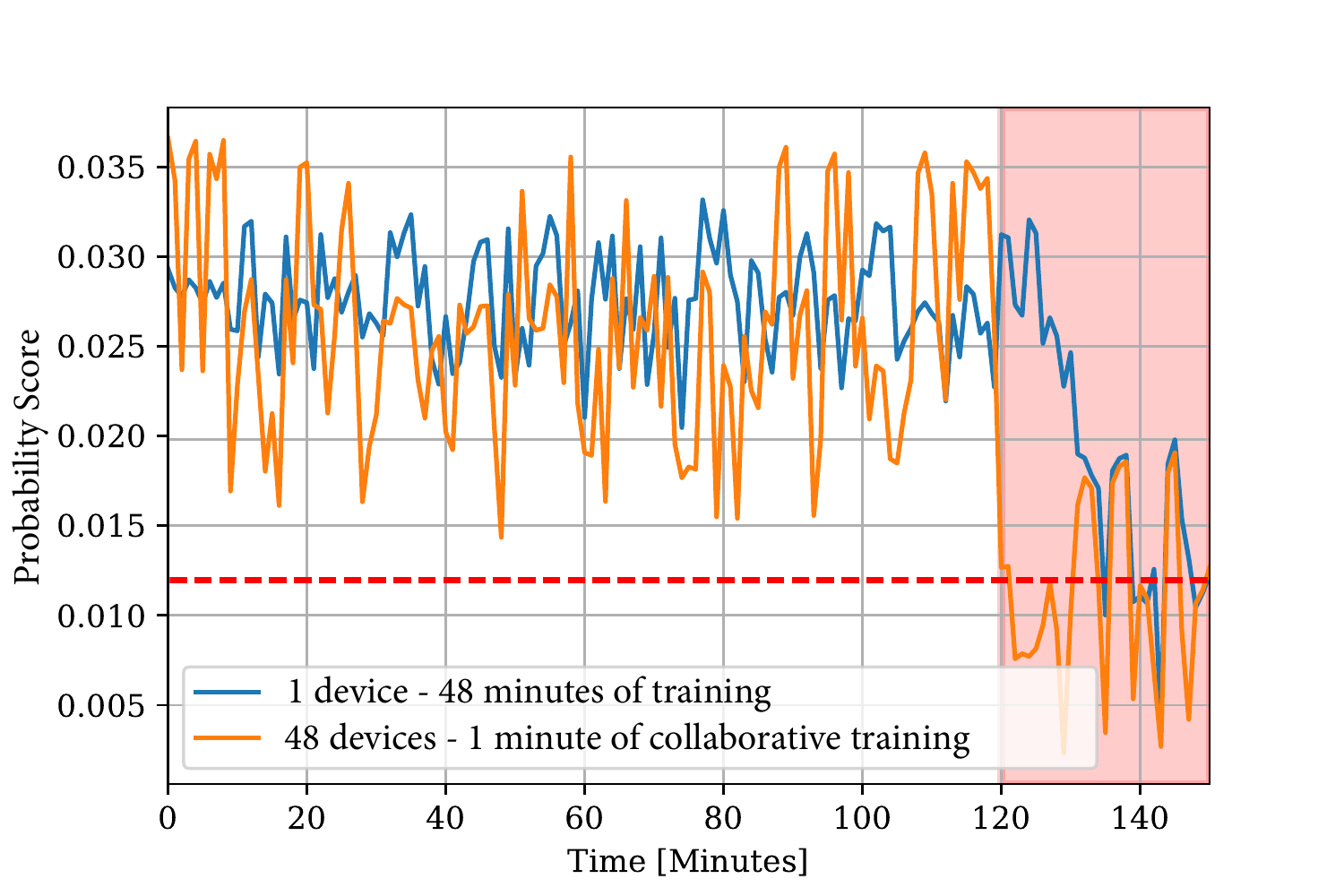}
	\label{Collaborative_Training_Exp:a}

	\caption{The probability scores of $M^{(\ell)}$ with 48 minutes of training, and $M^{(g)}$ with one minute of training across 48 devices (Exp1).}
	\vspace{.3cm}
	\centering
	\includegraphics[width=.8\columnwidth]{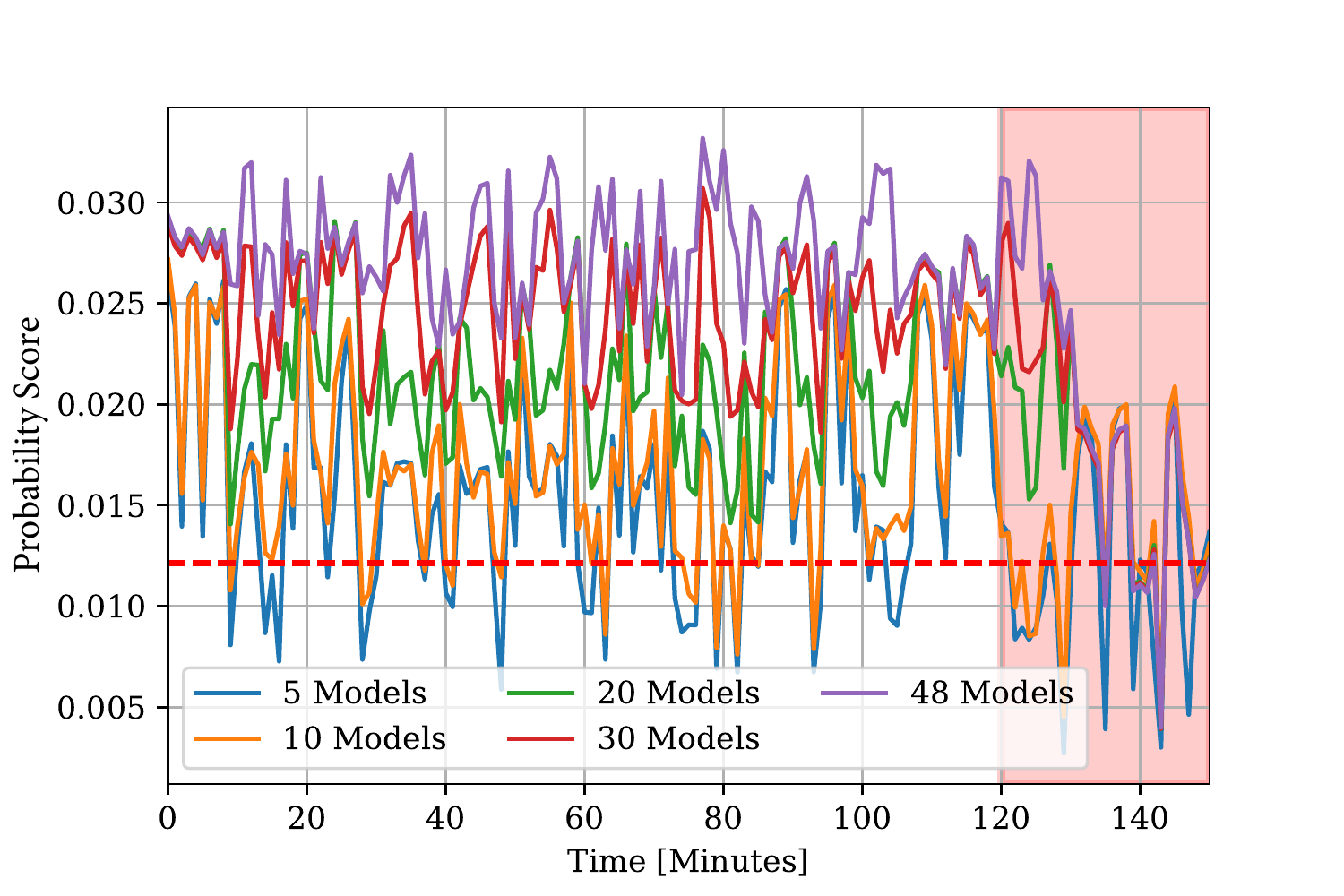}
	\label{Collaborative_Training_Exp:b}

	\caption{The probability scores of $M^{(g)}$ with increasingly larger sets of models (devices) in the case of Exp1.}
\end{figure}

In Table \ref{Collaborative_Training_ExpSummery}, we present the false positive rates (false alarm rates) of the framework with various numbers of devices and train time. The Table shows that just 48 devices training for two hours (2 days of experience) is enough to mitigate the false alarms. For the code-reuse and buffer-overflow attacks, there were no false negatives. However, in the replay-attack (Key-Reinstallation) there were a few false negatives. However, since the attacker sends a malformed packet multiple times, we ultimately detect the attack.

\subsubsection{Resilience Against Adversarial Attacks}
Since agents are constantly learning (even after $T_{grace}$), it is important that the framework be resilient against accidentally learning malicious behaviors as benign (i.e., poisoning). The acceptance criteria of a partial block ensures that these behaviors are not incorporated into the global models.

If some of the IoT devices are infected after the publication of the first block, we expect the collaborated $N^{(g)}$ to detect the malware, and not learn from it by accident. However, let's say that some of the IoT devices were infected prior to the publication of the first block and the elapses of $T_{grace}$. When the infected agents add their poisoned model to the \textit{partial-block}, other poisoned agents will reject their \textit{partial-blocks} because the \textit{model-attestation} step will reveal that the potential new $N^{(g)}$ is very different than their own local models $N^{(\ell)}$. Fig. \ref{Linear_Distance} visualizes this concept as heat maps, where the intensity of index $(i,j)$ represents the linear distance between the probabilities of transition $M^{(\ell)}_{ij}$ and $M^{*}_{ij}$, where $M^{*}_{ij}$ is a combined model from a \textit{partial-block}. In \ref{Linear_Distance:a}, the \textit{partial-block} has $10$ clean models, and in \ref{Linear_Distance:a}, the \textit{partial-block} has $10$ poisoned models. When an agent performs \textit{model-attestation}, the agent will find that $d(N^{(\ell)},N^*)<\alpha$, and reject the \textit{partial-block}. Assuming $L$ is large enough (e.g., $L=10,000$), and that a minority of agents are not infected, we expect that a poisoned \textit{partial-block} will never be closed before a clean one achieves consensus. 

Let's say that $\alpha$ was set too low, or that the malicious jump sequences were very similar to the legitimate ones. In this case, the \textit{model-attestation} step will accept the \textit{partial-block}, but the \textit{abnormality-filtration} step will remove the malicious behaviors. This is assuming that less than $p_a$ percent of the models in the \textit{partial-block} contain the malicious transitions. Fig. \ref{Adversarial_evaluation} shows that with $L=20$ and $p_a=75\%$, an attacker must poison $15/20$ models (during $T_{grace}$) in order to evade the detection of the next $M^{(g)}$. This is very difficult for the attacker to achieve because (1) he must infect the IoT devices without detection, (2) there is a chance that not all infected models will appear together in a \textit{partial-block} (e.g., with 48 or 1,000 devices), and (3) if he does not succeed before the first block if published, then it is likely that the new $M^{(g)}$, accepted among \textit{all} agents, will detect the malware. 

Another possibility is that the attacker may try and sabotage the agent via target application. However, by accessing the agent's memory from the monitored application will require additional exploits from the malware. Ultimately, the agent will detect either the initial intrusion, or the exploits used to gain access to the agent's memory space.

Another insight is that when a minority of models are infected yet the agent's \textit{model-attestation} accepted the \textit{partial-block}, the \textit{abnormality-filtration} removes the malicious transitions but keeps the benign ones (observed by $p_a$ percent of the models). As a result, healthy information is retained from the poisoned models, while the abnormalities are filtered out.

\begin{figure}[p]
	\centering{
		\subfloat[Linear distance between a benign model and a clean combined model] 
		{\includegraphics[width=.8\columnwidth]{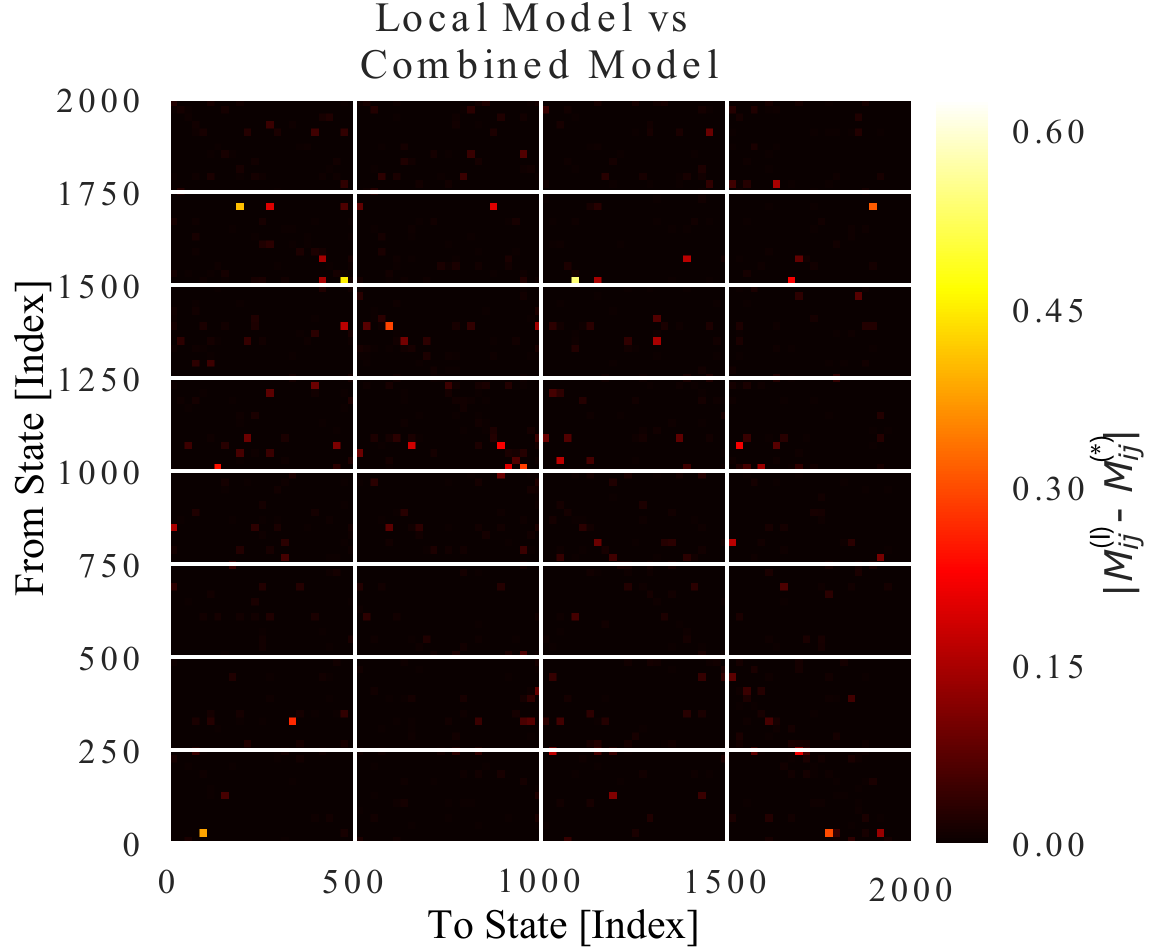}
			\label{Linear_Distance:a}
		}\quad
	
		\subfloat[Linear distance between a benign model and a positioned combined model.]
		{\includegraphics[width=.8\columnwidth]{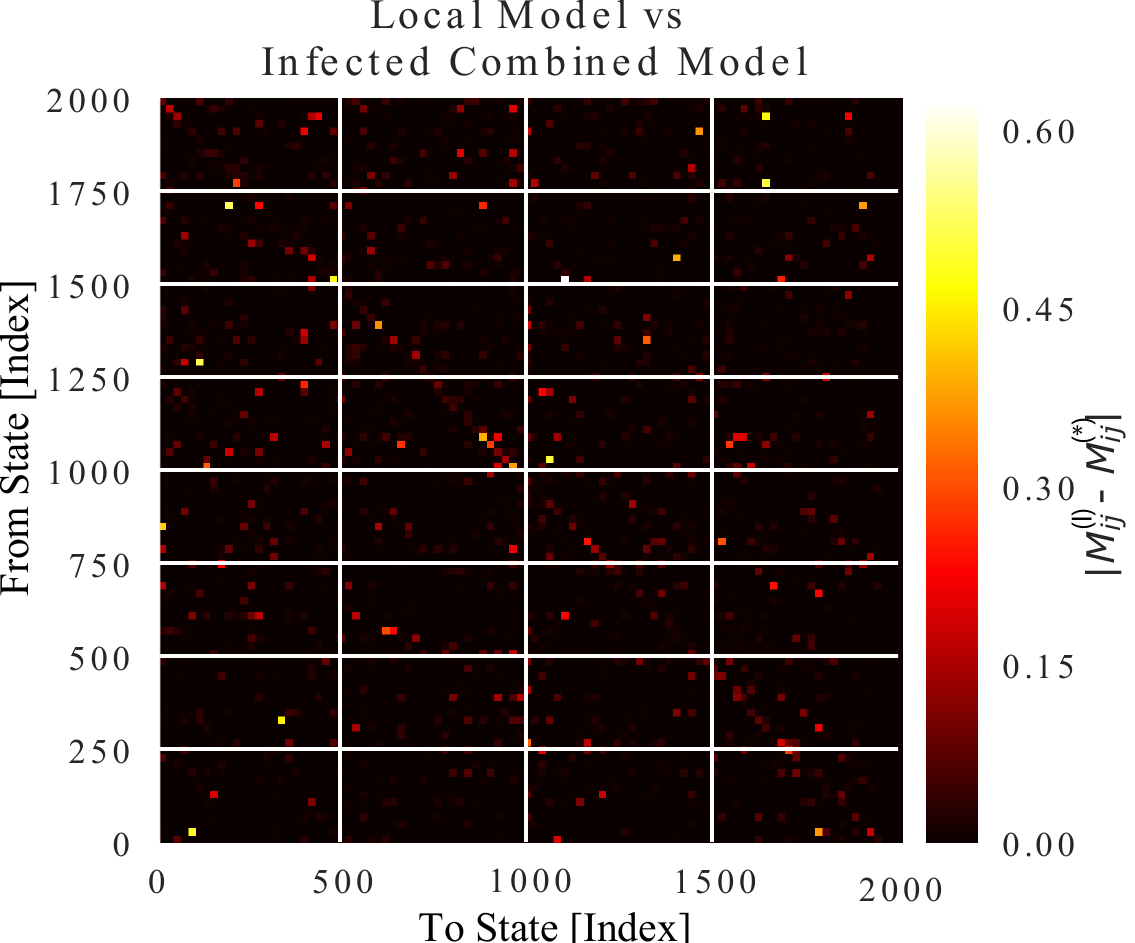}
			\label{Linear_Distance:b}
		}
	}
	\caption{Heat maps of the linear distance between models in Exp4.}
	\label{Linear_Distance}
\end{figure}
\begin{figure}[h]
	
	\centering{
		\includegraphics[width=.8\columnwidth]{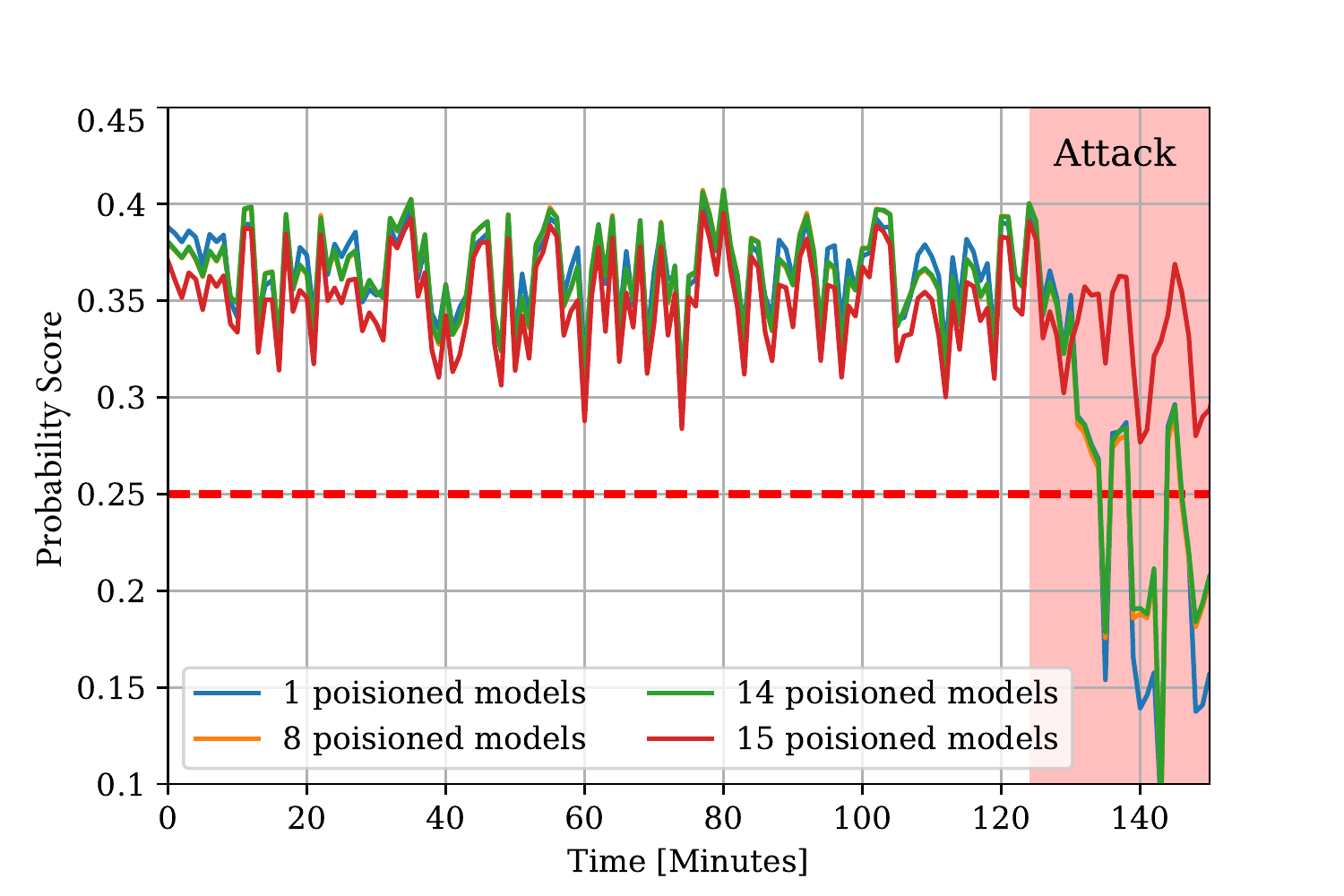}
	}
	\vspace{-0.5cm}
	\caption{The combined model normalized probability generated from the latest block $B$, where various numbers of the models in $B$ have been infected (attacked).}
	\label{Adversarial_evaluation}
\end{figure}

\subsubsection{Baseline Comparisons}
To understand the capabilities of the proposed collaborative framework, we evaluate the selected the anomaly detection method (EMM over memory regions) and the entire host-based intrusion detection system (the blockchain framework) to their respective baselines. 

To validate the use of the EMM, we compare its performance to two well-known sequence-based anomaly detection algorithms: t-STIDE and PST (see \ref{sec:relworks}). For the PST we took a sequence length of 10. We also compare the EMM to the heatmap method proposed in \cite{7167219}. In these experiments, we performed the buffer overflow attack in the Smart Light (Exp1), the code reuse attack on the Smart Camera (Exp4), and the replay attack on the router (Exp5). All of the algorithms were given the same 30 min of normal training data and then were tested on 20 min of normal data followed by 10 min of attacks.

To measure the performance we compute the area under the curve (AUC). The AUC is computed by plotting the true positive and false positive rates (TPR and FPR) for every possible threshold, and then by computing the area under the resulting curve. Intuitively, it provides a single measure for how well a classifier performs. A value of `1' indicates a perfect predictor and a value of `0.5' indicates that the predictor is guessing labels at random. Since the AUC measure ignores precision it is slightly misleading in the case of anomaly detection. Therefore, we also compute the average precision-recall curve (avPRC) which is computed in a similar manner.

In Fig. \ref{fig:aucprc} we present the results from this baseline test. We found that although t-STIDE sometimes our performed the MC, the MC consistently provide the best performance for all target applications. This justifies our use of the EMM for our system. We also note that the PST took several hours to train on a strong PC, and therefore is not practical to train on an IoT.

\begin{figure}[h]
	\centering
	\includegraphics[height=.25\textheight]{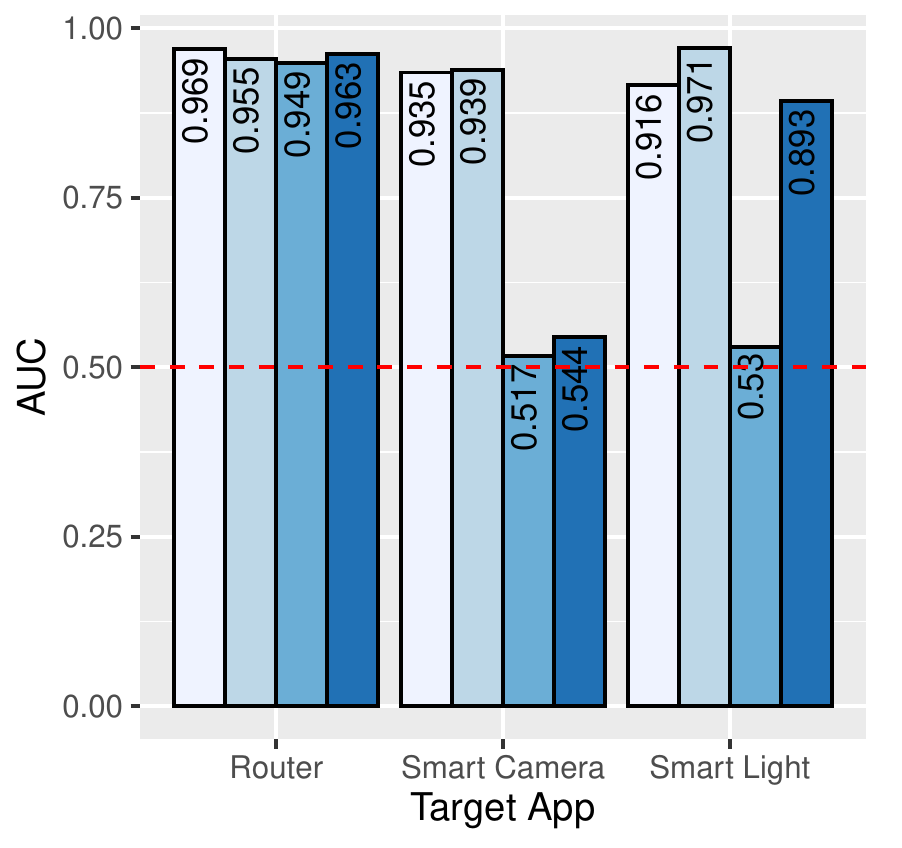}	\includegraphics[height=.25\textheight]{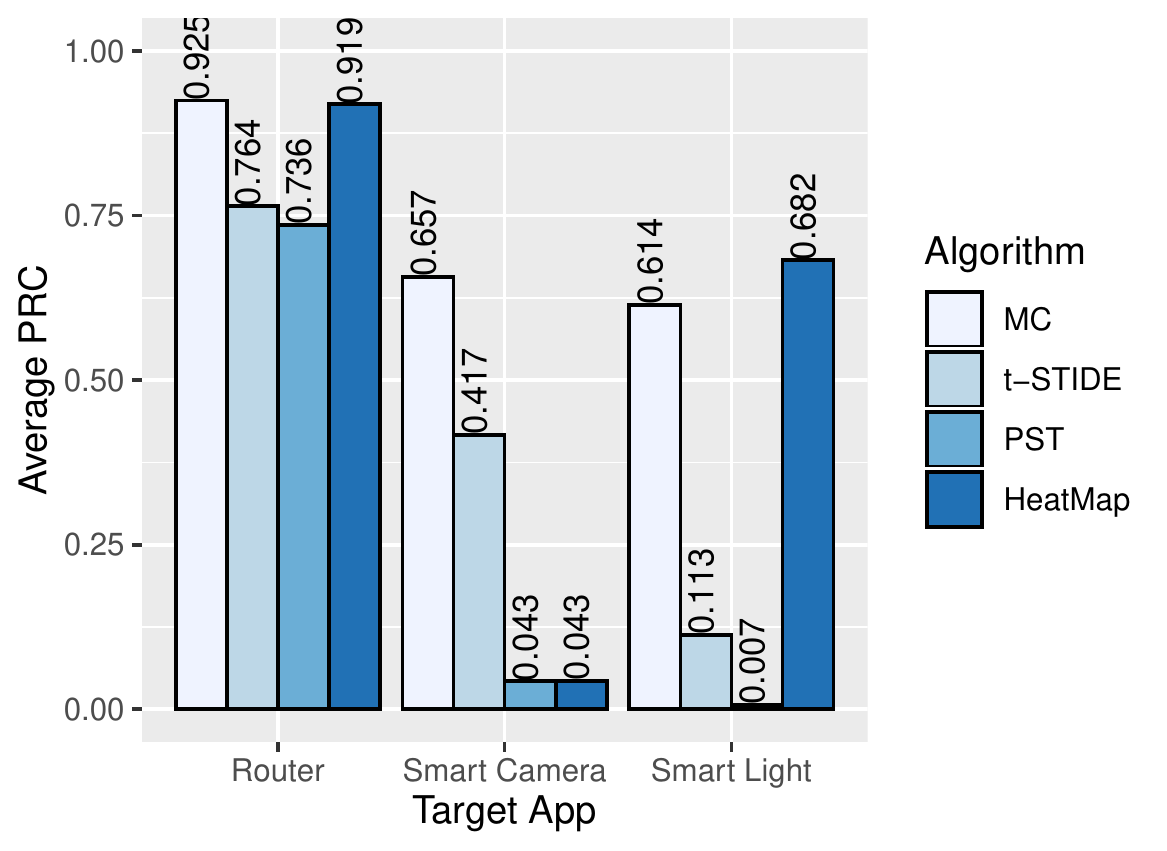}
	\caption{The AUC (left) and the PRC (right) of each algorithms for each attack/target app.}
	\label{fig:aucprc}
\end{figure}

In Fig. \ref{fig:score_time} we plot the anomaly scores (predicted probabilities) of the algorithms over time during the attack phase. From the figure, it is clear why the MC consistently had a high avPRC since there is a clear separation between the anomalous scores and benign scores. This is important when deciding on a threshold. In practice, the threshold is determined based on a statistical measure given the benign data distribution.

\begin{sidewaysfigure}[p]
	\centering
	\includegraphics[width=\textwidth]{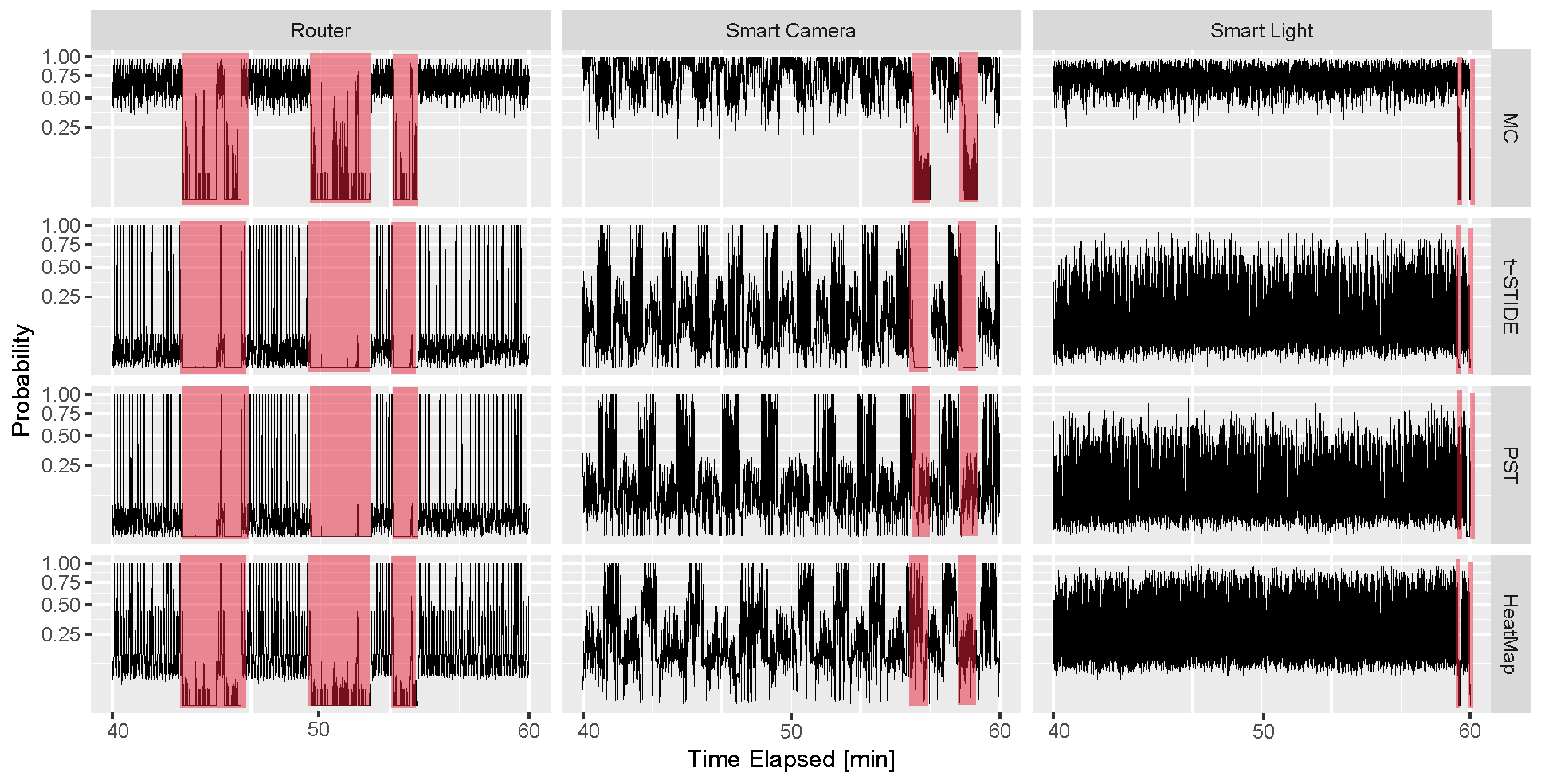}
	\caption{The anomaly scores (predicted probabilities) of the algorithms over time during the attack phase. The actual attack periods are marked in red.}
	\label{fig:score_time}
\end{sidewaysfigure}

To validate the use of entire framework, we evaluated our host-based intrusion detection system (H-IDS) in comparison to others. Since we targeting IoT devices, we selected H-IDSs which are well-known, operate on Linux, and can be compiled to run on an ARM processor: OSSEC, SAGAN, Samhain, and ClamAV. OSSEC is an open-source system which performs integrity checking, log analysis, rootkit detection, time-based alerting, and active response. We loaded OSSEC with all it's default detection rules. SAGAN is an open source multi-threaded system which performs real-time log analysis with a correlation engine. Sagan's structure and rules work similarly to the Sourcefire Snort IDS/IPS, and we loaded it will all available community rules. 
Table \ref{tab:hids} compares the H-IDSs to ours in the context of the content being monitored, and the intrusion detection mechanism used. Samhain is an integrity checker and host intrusion detection system. Finally, ClamAV is a free software open-source antivirus software which we loaded will all current virus signatures.

Once we loaded all four H-IDSs onto a Raspberry Pi, we launched each target application and performed the same attacks described above. We found that none of the four H-IDSs reported any alerts. This makes sense because these systems do not perform dynamic analysis on the target application's control flow. Therefore, the buffer overflow, code reuse, and replay attacks evaded detection.

\begin{table}[!t]
	\begin{center}
		\caption{The Host-based Intrusion Detection Systems compared to Ours}
		\vspace{1em}
		\begin{tabular}{c}
			\includegraphics[width=.5\columnwidth]{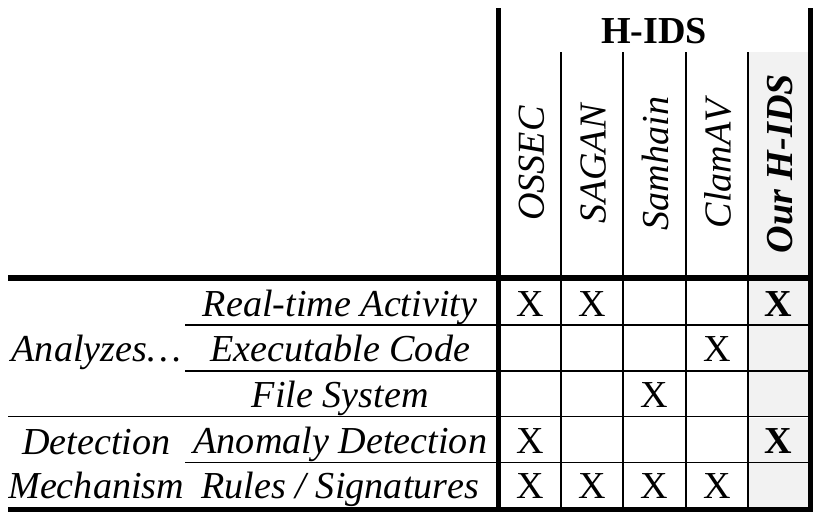}
		\end{tabular}
		\label{tab:hids}
	\end{center}
\end{table}

\begin{figure}[!t]
	\centering
	\includegraphics[width=.9\textwidth]{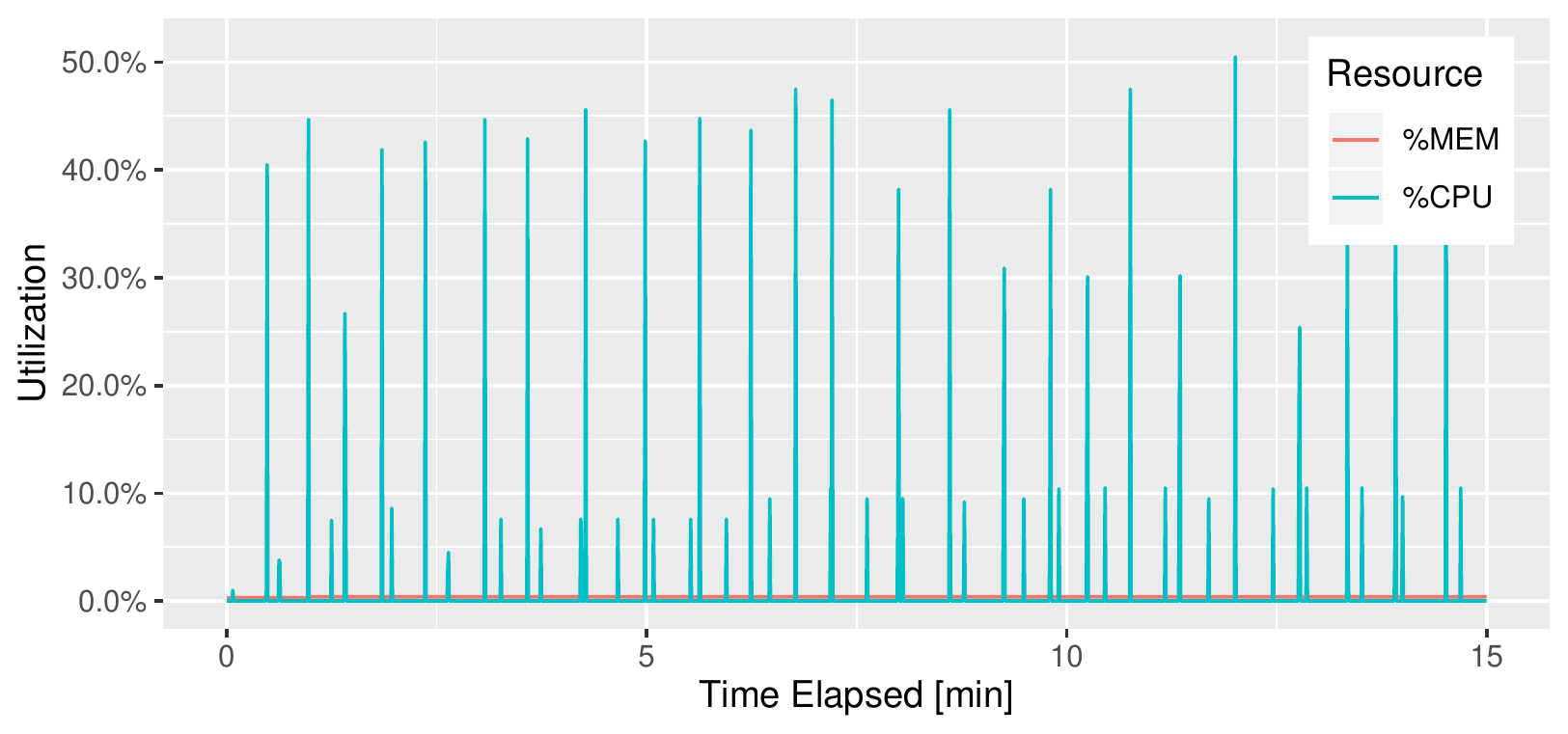}
	\includegraphics[width=0.48\textwidth]{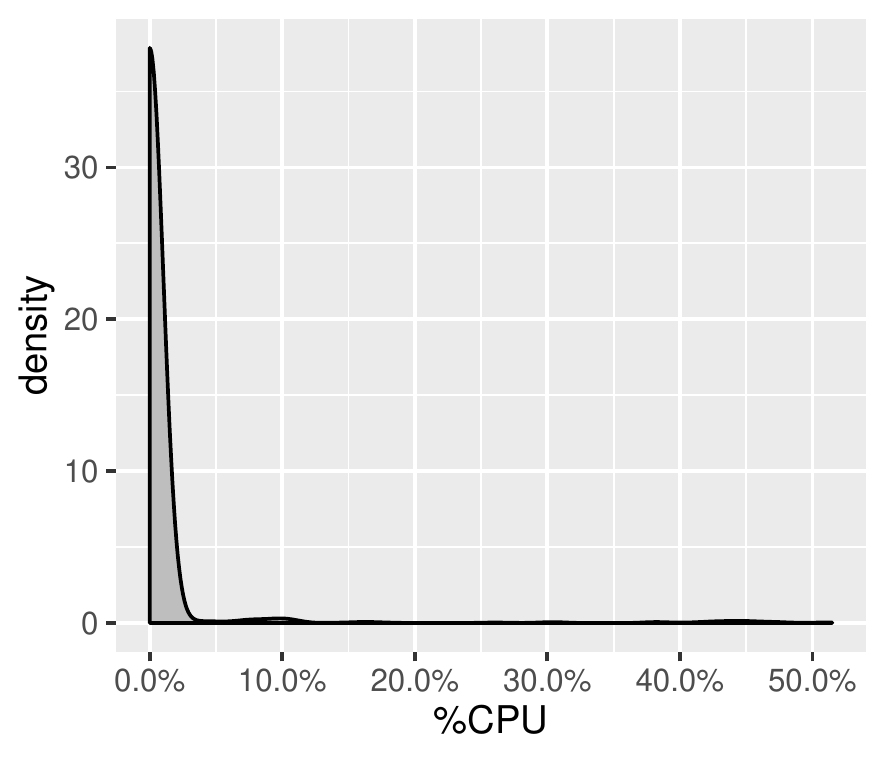}
	\includegraphics[width=0.48\textwidth]{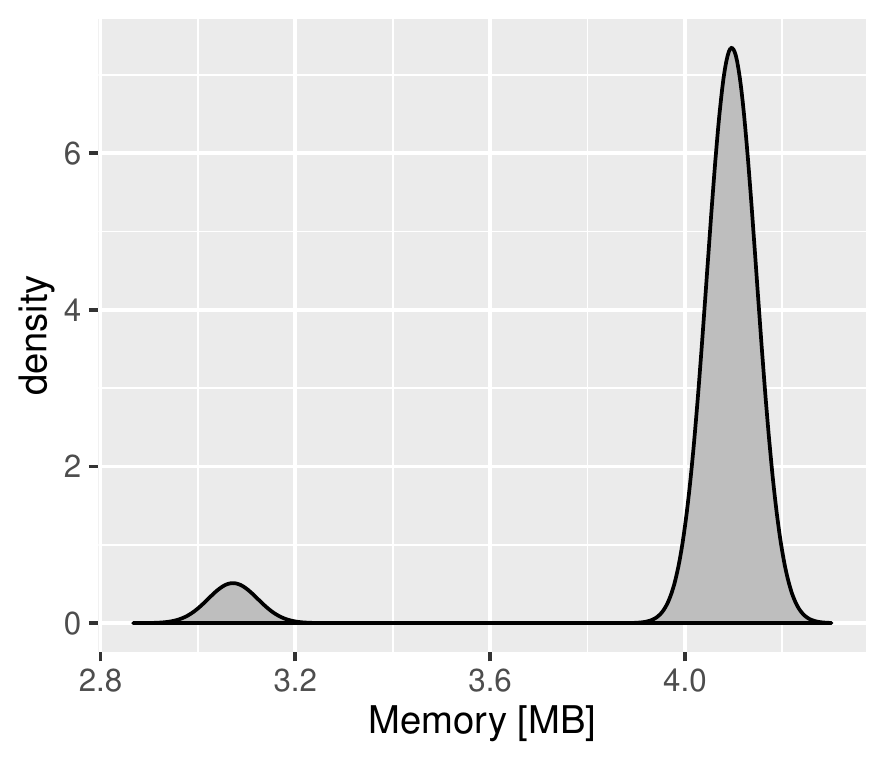}
	\caption{The resource utilization of the agent on a 500 MHz CPU. Top: Resource utilization over the first 15 minutes. Bottom: Resource utilization expressed as density plots.}
	\vspace{-0.3cm}
	\label{fig:benchmark}
\end{figure}

\subsection{Complexity Analysis \& Benchmark}\label{sec:complexity}
The time complexity of an agent can be broken down according to the three parallel processes in Fig. \ref{LLD}. The Gather Intelligence process periodically receives a ring buffer from the kernel will the last $n$ jump operations, checks for anomalies, and updates the EMM. Therefore, It's complexity is $O(n)$. However, if an averaging window is used over the anomaly cores, then the complexity is $O(n+wn)$ where $w$ is the window size. However, $w<<n$ in practice and thus we can consider the complexity to remain as $O(n)$.  
When the Retrieve Intelligence process receives a loner chain than the local one, it will check the legitimacy of the last block by validating its signatures, and then possible validate the signatures int he partial block as well. Therefore, in the worst case scenario, the agent will perform $2L-1$ signature checks. Although this may take a second to process, it will not affect the system since at most $b$ broadcasts will be accepted by the agent during each interval $T$, where $T$ in the order of minutes or greater.
Finally, the Share Intelligence process wakes up and sends the local chain to $b$ other agents at random. Although the p2p discovery protocol and network transfer may take some time, it has a negligible affect on the CPU.

We performed a performance benchmark to evaluate the agent's CPU and memory utilization. The benchmark was performed on a Linux embedded device with a single ARM Cortex CPU clocked at 500 MHz since this CPU configuration is common among IoT devices \cite{ARM}. The test was run for one hour in the presence of 48 other agents having the same protocol configurations used in the evaluations. The target app was a web facing log server with a known CVE. The results can be found in Fig. \ref{fig:benchmark}. The results show that the resource consumption of the agent is negligible using only 1\% of the CPU on average and 4MB of RAM.

\subsection{Blockchain Simulator}\label{subsec:simulator}
To help other reproduce our work and understand how the blockchain protocol works, we have \href{https://drive.google.com/drive/folders/15gLytEJyQyYCmhB-EZSkES77KsuCW0hw?usp=sharing}{published a discreet event simulator (DES)} of the protocol written in Python.\footnote{\texttt{https://github.com/ymirsky/CIoTA-Sim}} The DES is object oriented and creates an instance of each agent to help users follow the protocol logic and the propagation of the chains. The DES only simulates the high level protocol logic (e.g., Section \ref{subsec:deadlocks}), and not the \textit{model-attestation} training or combining. For code on model management, please see our other repository.\footnote{\texttt{https://git.io/vAIvd}}

The user selects the number of agents, $L$, $T$, the number of blocks to close, and the connectivity between the agents. The connectivity can be set to fully connected or random: Barabasi-Albert Algorithm (preferential attachment) or Watts-Strogatz (small world attachment). The DES queue then manages the agent's information sharing (when elapses of $T$) where some small amount of noise is added to the event times. 

For each type of graph, we ran the simulator 100k times with 1,000 agents and set $L=800$. For the Barabasi-Albert generator we set attachment to 1, and for the Watts-Strogatz generator we set the neighbors to 5 with a probability of 0.1. For each trial we generated a new random network. 

Table \ref{tab:sim} presents the agent connectivity (node degree) of the agents in the simulations and the number of times (epochs) an agent executed step \ref{step:broadcast} of the protocol until a block was closed. Fig. \ref{fig:sim} plots the distribution of the epoch counts over 100k trials. 
The results show that blocks are closed faster with better connectivity between agents (larger node degrees). A fully connected network closes a block in 1 epoch and a sparsely connected network (Barabasi-Albert) can take up to 800 epochs (2.2 hours with $T=10$ seconds).


\begin{table}[!t]
	\resizebox{\textwidth}{!}{
	\begin{tabular}{@{}lccccc|cl@{}}
		& \multicolumn{5}{c|}{Node Degree} & \multicolumn{2}{c}{\#Epochs} \\
		Graph Generator & Min & Max & Median & \textbf{Mean} & Std. & Mean & Std. \\ \midrule\midrule
		Complete  & 999 & 999 & 999 & \textbf{999} & 0.00 & 1 & 0 \\
		\textit{agents connected to all agents} &&&&&&&\\ \midrule
		Watts-Strogatz & 6 & 10 & 6 &\textbf{6.59} & 0.74 & 142.80 & 3.14 \\
		\textit{small world attachment} &&&&&&&\\ \midrule
		Barabasi-Albert & 1 & 99 & 1 & \textbf{1.99} & 3.56 & 800.46 & 3.16 \\
		\textit{preferential attachment} &&&&&&&\\ 
		 \bottomrule\bottomrule
	\end{tabular}
}
	\caption{The network generators' node degrees (number of neighbors per agent), and the number of epochs ($T$) elapsed until a block was completed.}
	\label{tab:sim}
\end{table}

\begin{figure}[!t]
	\centering
	\includegraphics[width=0.5\textwidth]{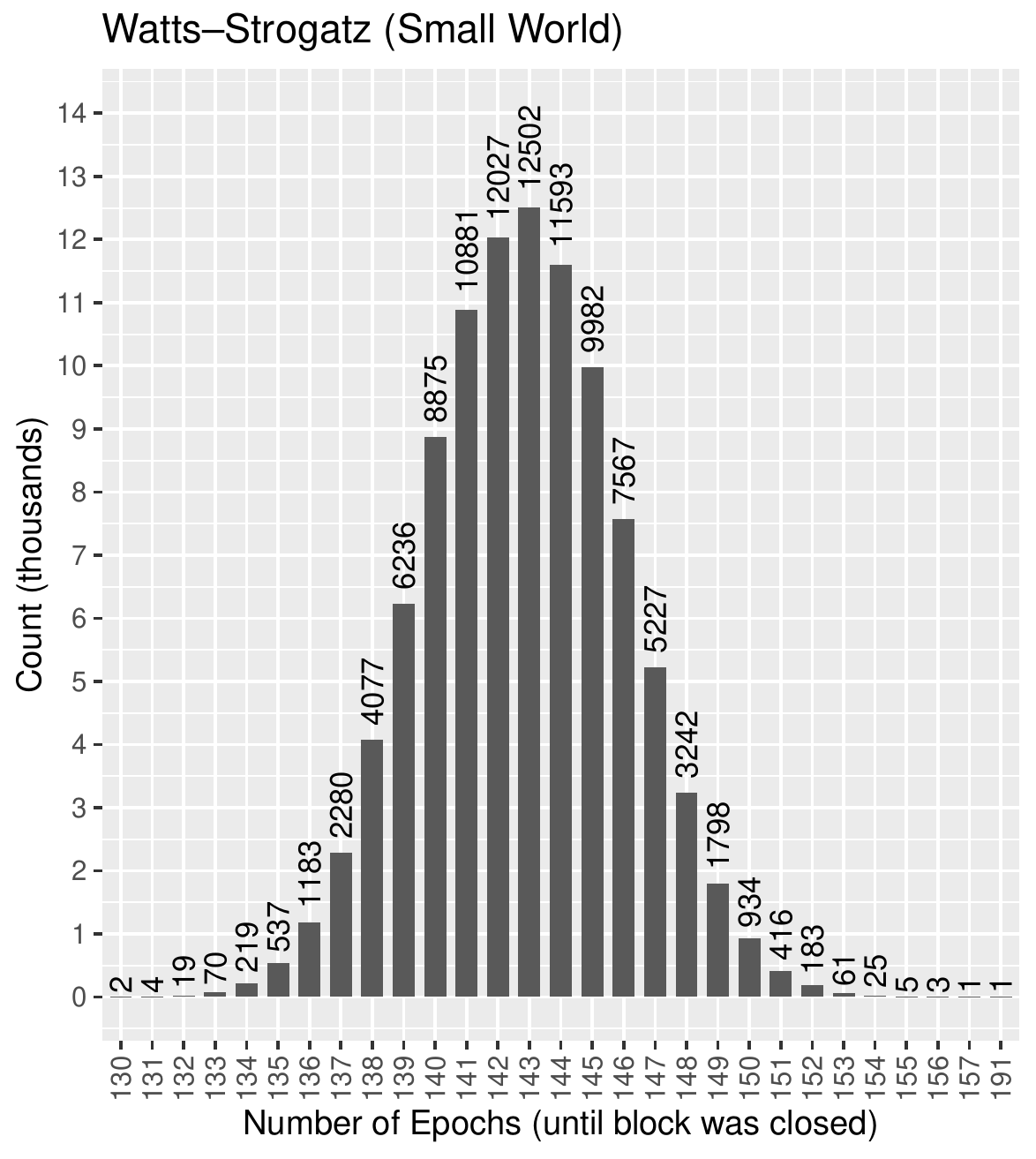}
	\hspace{-.7em}
	\includegraphics[width=0.5\textwidth]{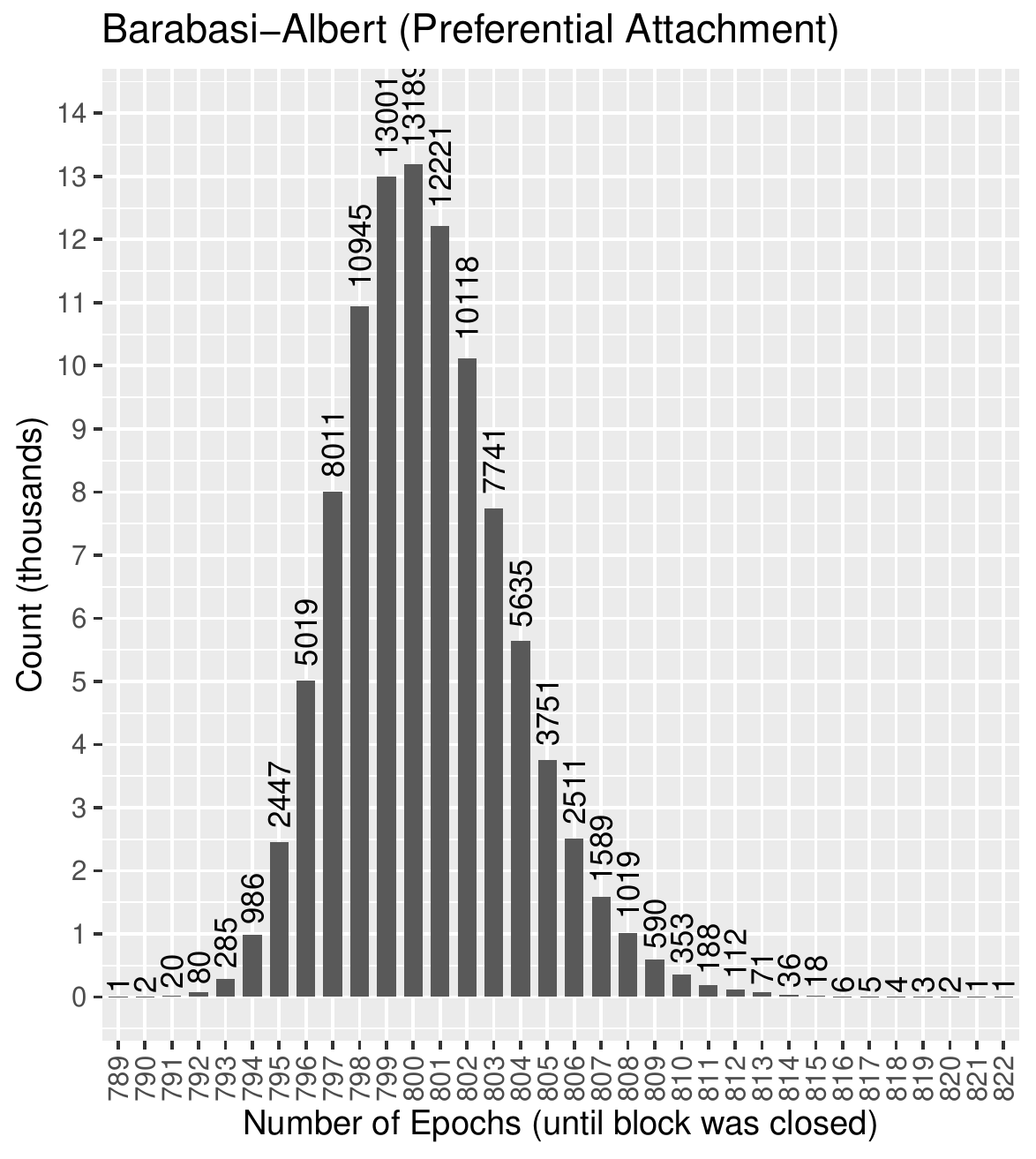}
	\caption{The number of epochs (times $T$ elapsed) until a block was completed, sampled one million times from each type of random network (containing 10,000 agents).}\label{fig:sim}
\end{figure}
\vspace{1em}

\section{Security Analysis}\label{sec:security}
In this section we will discuss the security coverage of the agents and potential attacks against the framework.

\subsection{Agent Coverage}
It is not necessary to have an agent monitor every application on an IoT device because the most common attack vector on the IoT is via the devices' Internet-facing applications. Therefore, to provide maximal coverage, we recommend that an agent protect all such applications (web servers, telnet daemons, ...) 

The agent can detect the exploitation of software vulnerabilities, misuse, and denial of service attacks if the attacks result in an abnormal control-flow, or an irregular read/write operation in memory:
\begin{description}
	\item[Confidentiality] There are many ways in which an attacker can violate the confidentiality of the device. Our agent can only detect these attacks if it causes an irregular control flow. For example, if the attacker performs an algorithm downgrade attack on an encryption channel like in the cases of Beast, Poodle, and Krack \cite{vanhoef2017key,moller2014poodle,sarkar2013attacks}.  then the protocol behaves differently than usual (as shown in our evaluation). Another example is where the attacker pulls data from a database via SQL an injection. In this case the interpreter will perform two queries back-to-back or ti may even perform a query that is never performed (e.g., drop table). Lastly, the attacker may perform a buffer overflow to reveal data in the server as was done in heartbleed to obtain private encryption keys \cite{durumeric2014matter}. When a buffer-overflow occurs, the program counter moved operates in a region of memory in a transition that it never does. 
	\item[Integrity] In some cases, an attacker may want to compromise the integrity of an IoT device by executing custom code or by abusing existing code (reuse). In doing so the attacker could install malware to recruit the device into a botnet, perform an act od ransomware, or some other malicious act. Many, of these code executions are achieved by altering the memory control flow. For example, buffer-overflow attacks via web inputs and irregular interactions with services (e.g., using ssh/telnet to download a  payload \cite{antonakakis2017understanding}). Moreover, in some cases the agent can discover when an attacker is perform reconnaissance to reveal potential vulnerabilities. One approach is to brute-force well-known credentials and another approach is to execute a variety of crafted malformed inputs until one succeeds (e.g., directory traversal attacks). In both cases, the operations will generate irregular loops in memory.
	\item[Availability] The goal of an attacker may be to disable the device's web server in a denial of service attack (e.g., if the device is a surveillance camera). Our agent can detect and alert for some of these attacks. For example, the agent can detect when a malformed packet causes a server to halt, and when SYN flood or SSL Renegotiation attack occur (due to the looping). However, the agent cannot detect a loss in connectivity if the channel is jammed, overloaded, of if a nearby router has been compromised. 
\end{description}
Please refer to Fig. \ref{fig:cwe} in the appendix for a list of high-level Common Weakness Enumerations (CWE) which the agent covers. We note that the agent's performance is not perfect, and the detection of some of these attacks may require a smaller region (state) size to provide the right granularity. In future work we plan to address this issue by letting he agent choose the state size based on the amount of code loaded into memory.

\subsection{Adversarial Attacks}\label{subec:adversarial}
Although the framework aims to protect IoT devices from attackers, to ensure reliability, we must consider how an attacker may target our framework. Once again, we will refer to the `CIA' of security:

\begin{description}
	\item[Confidentiality] One potential attack against the system is to intercept and extract $M^{(\ell)}$ for a target device. In doing so, the attacker may be able to violate the user's privacy by inferring the user's \textit{high-level} interactions with the application. This threat only applied during the creation of the first block since the shared models have not yet been generalized (combined) with others. To mitigate this threat, the manufacturer can initiate the blockchain with an initial model, similar to how stream ciphers use IVs (see \textit{cold start} later in section \ref{sec:discussion}). It is also a good idea to choose a relatively large state size for better obscurity. Another concern is if the user installs 3rd party apps, and the manufacturer has an agent automatically assigned each new app. In this case the attacker can infer which apps the user is using and when by monitoring the broadcasted chains. To mitigate this threat and the other cases, all p2p communications should be encrypted using SLL.
	
	\item[Integrity] If an attacker corrupts (poisons) the model in training, he can intentionally cause high false alarm rates or evade detection. There are two ways in which an attacker can poison the model. The first way is to install malware on the majority of the population before the first block is closed (supply chain attack or a regular infection). Infecting the majority of devices is very hard to accomplish because (1) it involves a short time window, (2) requires infection without detection,\footnote{The detection phase begins after $T_{grace}$ and not after completing the first block} and (3) a large number ($p_a\%$) of devices need to be infected. While it is not impossible, it is considerably more difficult than exploiting a single device. To minimize the attack window and bandwidth of the system, one may change the protocol so that $T$ is a monotonically increasing with the length of the local chain. For example, a linear function can be used or the exponential decay function $T(m)=(1-2^{-\lambda m})*(t_{\text{max}}-t_{\text{min}})+t_{\text{min}}$, where $m$ is the length of the local chain, $\lambda$ is the half-life rate, $t_{\text{min}}$ is the shortest interval, and $t_{\text{max}}$ is the longest. 
	
	Another way to corrupt the model is to (1) evade detection by compromising the device via an application which is not monitored, or via physical access to the device, then (2) achieve root privileges to compromise the agent, then (3) repeat this process until the majority of agents are under the attacker's control, and then (4) broadcast a compromised model in unison to all other devices. To mitigate this attack, users should consider which apps should be monitored and the physical security of their devices. Another option is to place the agent in the device's TrustZone (e.g., \cite{TrustZon14:online}).
	
	We note that an attacker may attempt to avoid detection by crafting exploit to follow a common flow through the application's memory. However, designing such an attack very limited and difficult since the operations are limited to a normal jump sequence. Moreover, to initiate this flow, the attacker will need to initially override some instructions (e.g., buffer-overflow), and as a result will likely trigger an alert.
	
	\item[Availability] An attacker may attempt a denial of service (DoS) attack to overload the agent thus disabling the device, or to block an agent access to new models by disrupting the agent's connectivity. In our protocol (section \ref{sec:ciota}), we took steps to ensure that an agent cannot be overloaded with broadcasted chains by limiting the processing rate to $b/T$ chains per second (section \ref{subsec:peerdisc}). However, care should be taken regarding the software used in the agent's server implementation to avoid flooding attacks, buffer overflows, and other attack vectors. 
	
	Moreover, the system is robust to network and hardware failures. This is because the system is distributed and agents only need to be connected to $b$ other random other agents. Moreover, collaboration is only necessary to accelerate the EMM model's initial training, and to handle future concept drifts (e.g., software updates). Therefore, during an outage, each agent will still continuously (1) execute the latest model to detect attacks on the target application, and (2) update the local model $N^{(\ell)}$ on sequences which are considered safe (above the probability threshold). In this case, each agent will still act as an efficient standalone host-based anomaly detection system, and continue to improve its model until it converges.
	
\end{description}

\section{Discussion}\label{sec:discussion}
In this section, we discuss the assumptions and design considerations of framework.


\subsection{Assumptions}
The framework's primary goal is to autonomously learn more in less time.
To achieve this goal, we take the following assumptions.

\begin{description}
	\item[Population size] \textit{There are enough participants with the same hardware model to support the system.} As noted earlier, a separate block chain is maintained for each device model version. If the homogeneous population is not large enough, then consensus will never be reached. However, IoT products are often mass produced and is likely that there will be tens of thousands of identical devices deployed around the world at a given time. Once a block has been closed, future generations can benefit from it even if the population has decreased below the consensus threshold.

	\item[More is better] \textit{Learning from more data will produce a better anomaly detection model.}
	This is because more data captures a more complete view of the behaviors, and therefore, the trained model will have a lower the false positive rate (FPR).

	\item[Achieving consensus]\textit{Given an appropriate $\alpha$, the majority of participants will reach a consensus for $M^{(g)}$ among themselves.}
	With a very large $\alpha$, the agents will surely achieve consensus. Although smaller values of $\alpha$ will improve the quality of the consensus, it will also increase the likelihood that \textit{partial-blocks} will be rejected --increasing the time it takes to complete a block (achieve consensus). If $\alpha$ is too low, consensus may never converge, and a collaboration will never occur. In our evaluation and proof of concept, we used empirical observations to select a constant value for $\alpha$. 
	However, as future research, a dynamic algorithm for selecting the appropriate value of $\alpha$ should be used to optimize the quality-convergence trade off.	
	\item[Benign majority] \textit{The majority of local models distributed among the agents are not poisoned.} 
	The blockchain is a peer-to-peer (P2P) protocol for a distributed (server-less) database that is managed by consensus of the network.
	Therefore, the blockchain architecture can only work under the assumption that most of the participants are clean at the outset. In the case of our framework, this means that the majority of agents are not accidentally updating their local models with malicious behaviors.
	
	\item[Benign start] \textit{An agent on a device monitors all applications (on separate chains) which are potential infection vectors from the Internet.} If the IoT device gets infected via an application which is not being monitored, then the agent cannot detect the threat. Therefore, in order to protect the device, all applications which can be exploited via the Internet should be monitored. 
\end{description}

\subsection{Implementation}

The following are some discussion points which relate to the implementation of the framework.

\begin{description}
	\item[Remediation Policies] While detection is a powerful tool for security, without acting on detected threats, detection is meaningless. When an agent detects an anomaly, the agent can (1) send an alert to a control server, (2) suspend the infected application via the kernel, (3) restart the infected application, or (4) a  perform a combination of these options. We note that scope of this paper is detection, whereas remediation is a task specific problem. For example, restarting the infected app may be acceptable for a smart air conditioner, but not a survallaince camera (an implicit DoS attack). Therefore, the remediation should be considered accordingly.
	
	\item[Address Space Layout Randomization] Today, many operating systems use Address Space Layout Randomization (ASLR) \cite{shacham2004effectiveness} to prevent outside entities from knowing the memory layout of applications in execution. ASLR ensures that an identical application will have a different memory layout on each device, making it very hard for attackers to exploit memory corruption vulnerabilities. However, since ASLR is an internal state, each agent can parse the memory layout of other agents to its natural state. 
	Concretely, for each memory region captured by $N^{(\ell)}$, an agent will include a library identifier and the region's offset from the library's initial address for other agents to rearrange the model. 
	
	\item[Authentication and Identification] In the paper, the framework uses PKI (public-key infrastructure \cite{adams2005internet}) in order to prevent an attacker from creating or replaying fake chains or records. Conventional PKI uses certificate authority (CA) servers to sign, manage, verify and revoke public key certificates. The use of a CA may incur some delay in processing blocks. However, there is no immediate rush to process these blocks since $T$ is in minutes or hours, and the agent continues to perform real-time intrusion detection in the meantime. Regardless, using CAs introduces single points of failure since they are centralized. To overcome this, several researchers have developed PKI for IoT networks \cite{8537812,8611563,JIANG2019185,shetty2019blockchain}. For example, in \cite{8537812} the authors propose three different methods for distributing CAs over a blockchains using Etherium smartcontracts and even the Emercoin infrastructure.
	Aside from PKI, another option is to use a shared secret among the agents and perform symmetric encryption. By doing so, no additional infrastructure is needed.
	
	The risk of using symmetric encryption is that an attacker can obtain a device an extract the shared key and compromise all agents. Therefore, to implement this approach, we recommend using an IoT devices' TrustZone. A TrustZone is a safe house inside the device (untrusted territory) that has access to the untrusted territory within the device. In this setup, the agent and its symmetric encryption key (provided by the administrator) is located in the TrustZone. By doing so, the agent and it's secrets are secured while avoiding the issues of PKI. However, if the TrustZone does not implement tamper protection, an attacker can physically interact with the device to extract the key.
	
	\item[Memory Region Size (state size)] The memory region size is a parameter configured by the user. This parameter incurs a trade-off: a large region size has low false-positive-rate but a high false-negative-rate, and a small region size has a high true-positive-rate but a high false-alarm-rate. Although we found that $256$ Bytes is a sufficient size, one should consider finding the smallest region size possible for their application. Another option is to use a small region size but increase $T_{grace}$.

	\item[Cold-start] In this paper we presented how agents distributed across a set IoT devices can build a detection mechanism with no prior knowledge. Although we expect the collaboration process to help agents learn rare yet benign behaviors, some benign behaviors may never cross the $p_a\%$ threshold. For example, the function which is executed by a smoke detector when a smoke is detected. To ensure that these behaviors make it into $M^{(g)}$, a manufacturer can post $M^{(g)}$ as the starting point for the blockchain. In this case, $M^{(g)}$ is a model from a single device in a lab, which has been exposed to the rare functionalities. By bootstrapping the blockchain with an initial model, it is possible to maintain a secure population with very few devices since the initial window of exploitation is diminished. 
	
	We note that this can be implemented by reserving a special certificate for the manufacturer who can sign this block. Doing so would also benefit the manufacturer since he can force updates in cases where there are software updates that significantly affect the applications' behavior.
	
	In the cases where a cold-start is necessary, we stress that majority of the training should still occur on-site and not in the lab. The reasons are that (1) it costs less, (2) it is very difficult to simulate natural dynamic human interactions with the devices, (3) it is challenging to stimulate the sensors realistically, and (4) simulating every single possible control-flow (fuzzing) is not practical.

	\item[Deployment] There are two ways the framework can be deployed: open and closed. In an open deployment anyone can register an agent to the network. In this mode of operation, a central entity should be entrusted with registering new users in order prevent an attacker from registering many accounts and overtaking the consensus. In a closed deployment, only invited agents can participate in the blockchain consensus. This deployment prevents unwanted entities from corrupting or eavesdropping on the blockchain.
\end{description}

\section{Conclusion}\label{sec:conclusion}
The number IoT devices is steadily increasing. However, manufacturers seldom patch older models and unintentionally write vulnerable code. As a result, large numbers of IoT devices are being exploited on a daily basis. Due to the scale of the problem, a generic stand-alone method for monitoring and protecting these devices is necessary.
In this paper, we introduced a blockchain-based solution for autonomous collaborative anomaly detection among a large number of IoT devices. 

To detect the exploitation of software on an IoT device, an agent is deployed on the device and efficiently models the software's normal control-flow for anomaly detection. However, the model training is vulnerable to adversarial attacks, and it is unlikely that a single IoT device will observe all normal behaviors (sensors readings, triggers, interactions, etc$\ldots$). Therefore, the agent uses a blockchain protocol to incrementally update the anomaly detection model via self-attestation and consensus among other agents running on similar IoT devices. By collaborating among other agents, the training phase (convergence) is significantly shorter, and false-alarm-rate is reduced due to the shared experience.

To evaluate the proposed framework, we used 48 Raspberry Pis running a wide variety of IoT applications. We also made a discreet event simulator of the agents with different connectivity to simulate larger systems and help the reader follow the protocol. Our evaluations show that the proposed method can efficiency detect different types of attacks with no false-alarms (given enough devices and a sufficient training period). 

The proposed framework does not require any a manual process of creating virus signatures, or a manual process for pushing updates. Furthermore, IoT devices are able to detect exploits without prior knowledge of the exploits. In terms of practicality, the framework is platform generic, completely autonomous, and scales with the number of IoT devices. Therefore, the proposed framework has the potential to provide IoT manufactures with a cheap and effective solution. 

We hope that this framework, and its variants, will assist researchers and the IoT industry in securing the future of the Internet of Things.

\bibliography{CIoTA}{}

\Large{Appendix}
\begin{table}[t]
	\begin{center}
		\caption{The Common Weaknesses (CWE) which are covered by the proposed solution. Note that the agents ability to detect attacks on these weaknesses depends on the agent's configuration and the affect which the attack has on the control flow in memory.}
		\begin{tabular}{c}
			\includegraphics[width=.8\textwidth]{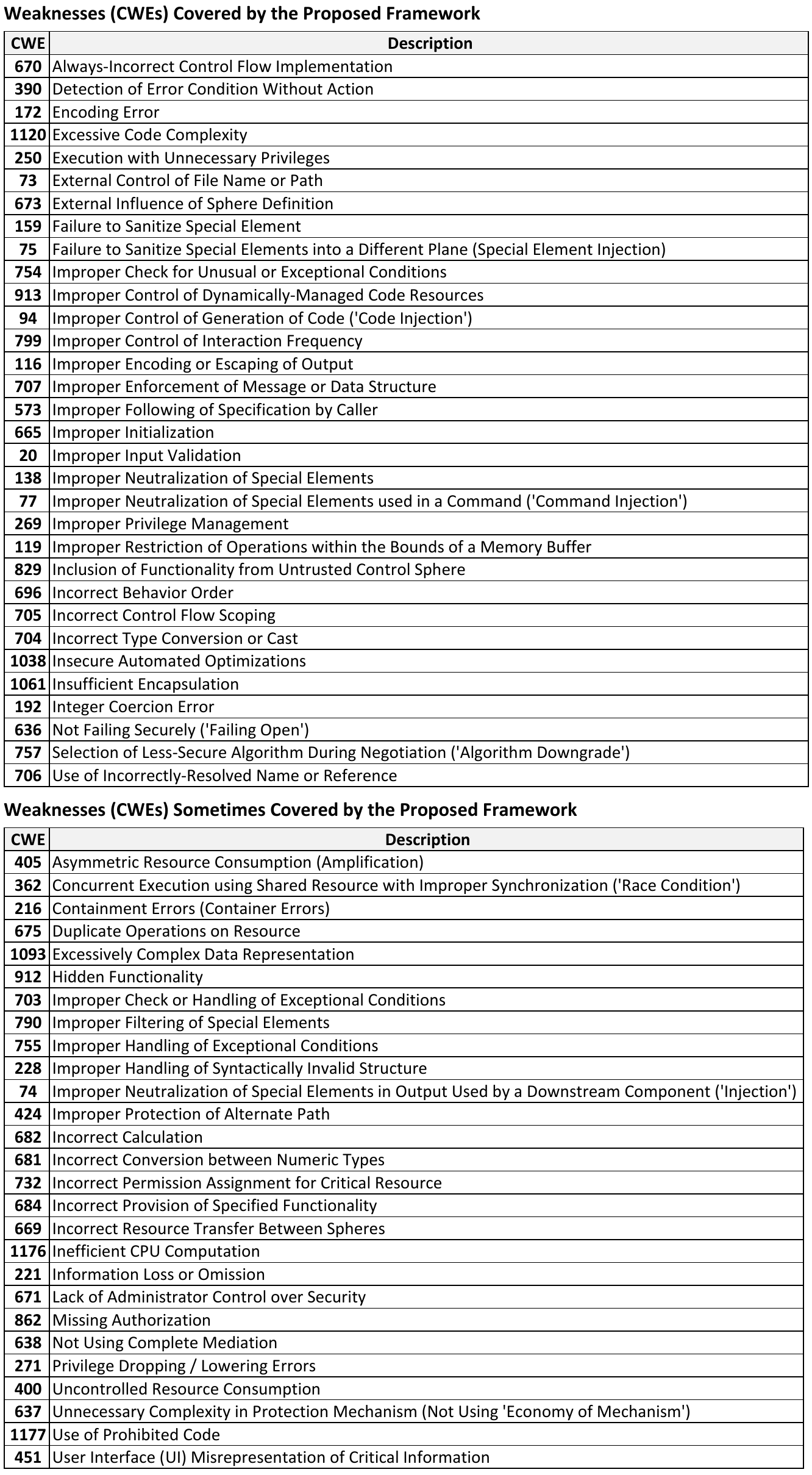}
		\end{tabular}
		\label{fig:cwe}
	\end{center}
\end{table}

\end{document}